%% file: root_MC.tex
\documentclass[journal,12pt,onecolumn,draftclsnofoot]{IEEEtran}

\usepackage{epsfig}
\usepackage{times}
\usepackage{thmtools, thm-restate}
\usepackage{float}
\usepackage{afterpage}
\usepackage{amsmath}
\usepackage{amstext}
\usepackage{amssymb,bm}
\usepackage{latexsym}
\usepackage{color}
\usepackage{graphicx}
\usepackage{amsmath}
\usepackage{amsthm}
\usepackage{graphicx}
\usepackage[center]{caption}
\usepackage{pstricks}
\usepackage{subfigure}
\usepackage{tikz}
\usepackage{booktabs}
\usepackage{multicol}
\usepackage{lipsum}
\usepackage{dblfloatfix}
\usepackage{rotating}

\usepackage{multirow}
\usepackage{tabularx}
\usepackage{pbox}

\allowdisplaybreaks

\newtheorem{thm}{Theorem}
\newtheorem{cor}[thm]{Corollary}
\newtheorem{lem}[thm]{Lemma}
\newtheorem{prop}[thm]{Proposition}
\newtheorem{rem}{Remark}
\newtheorem{defin}{Definition}

\begin{document}
	
\title{On Secure Network Coding for Multiple Unicast Traffic}
\author{Gaurav Kumar Agarwal,  Martina Cardone and Christina Fragouli}

\author{
	\IEEEauthorblockN{Gaurav Kumar Agarwal$^\star$, Martina Cardone$^\dagger$, Christina Fragouli$^\star$ } \\
	$^\star$ University of California Los Angeles, Los Angeles, CA 90095, USA\\
	Email: \{gauravagarwal, christina.fragouli\}@ucla.edu \\
	$^\dagger$ University of Minnesota, Minneapolis, MN 55404, USA, Email: cardo089@umn.edu
	\thanks{The results in this paper were presented in part at the 2016 IEEE International Symposium on Information Theory, at the 2016 IEEE Globecom Workshop, and at the 10th International Conference on Information Theoretic Security.}
}

\maketitle

\begin{abstract}
  \input{abstractMC.tex}
\end{abstract}

\section{Introduction}
\input{introductionMC.tex}


\section{Setup and problem formulation}
\input{setupMC.tex}

\section{Outer bound}
\label{sec:outer_bound}
\input{outer_boundMC.tex}

\section{Capacity achieving scheme for networks {with} two destinations}
\label{sec:two_destinations}
\input{achievability_two_destinationsMC.tex}

\section{Secure scheme for combination networks}
\label{sec:combination_nw}
\input{combination_networksMC.tex}

\section{Polynomial time scheme for networks with arbitrary topologies and arbitrary number of destinations}
\label{sec:two_phase_scheme}
\input{two_phase_scheme_for_generalMC.tex}

We now compare our two-phase scheme with the optimal scheme for networks with two destinations. We also analyze and discuss potential reasons behind the sub-optimality of the two-phase scheme.
\subsection{Complexity of the two phase scheme and a comparison with the optimal scheme}
\input{complexityMC.tex}

\section{{Comparisons, Non-reversibility and Additional Instances of Multiple Unicast Traffic}}
\label{sec:discussions}
{In this section, we conclude the paper with some comparisons and analysis of other instances of multiple unicast traffic. 
In particular, in Section~\ref{sec:CompUnsRate}, we compare the secure rate region for $m=2$ destinations in Theorem~\ref{thm:SecureLB2} with the capacity region when the adversary is absent. The goal of this analysis is to quantify the rate loss that is incurred to guarantee security. 
In Section~\ref{sec:non_reversibility}, we prove that the secure capacity region for $m=2$ destinations is non-reversible. Specifically, we show that, if we switch the role of the source and destinations and we reverse the directions of the edges, then the new secure capacity region differs from the original one.}
This is a surprising result since it implies that -- different from the unsecure case where {non-reversible} networks must necessary have non-linear network coding solutions~\cite{Koetter04networkcodes,riis2007reversible} -- under security constraints even networks with linear network coding solutions can be {non-reversible} if the traffic is multiple unicast.
{Finally, in Section~\ref{sec:OtherInst}, we consider other instances of multiple unicast traffic, such as networks with erasure links and noiseless networks with two sources and two destinations. For some specific network topologies, we derive the secure capacity region. This analysis sheds light on how coding should be performed across different unicast sessions.}

\subsection{Comparison with the Unsecure Capacity Region}
\input{single_source_unsecure_comparisionMC.tex}

\subsection{Non-Reversibility {of the Secure Capacity Region}}
\input{reversibilityMC.tex}

\subsection{{Analysis} on Other Instances of Multiple Unicast Traffic}
\input{other_instancesMC.tex}

\begin{appendices}
	\input{seperable_graph_m_2MC.tex}
 \input{decodingMC.tex}
\input{2CombNetAchMC.tex}
	\input{single_source_unsecureMC.tex}
\input{OtherNetInstMC.tex}
\end{appendices}

\bibliographystyle{IEEEtran}
\bibliography{tit18}

\end{document}

%% file: abstractMC.tex
This paper investigates the problem of secure communication in a wireline noiseless scenario where
a source wishes to communicate to a number of destinations in the presence of a passive external adversary. 
Different from the multicast scenario, where all destinations are interested in receiving the same message,
in this setting different destinations are interested in different messages. 
The main focus of this paper is on characterizing the secure capacity region, when the adversary has unbounded computational capabilities, but limited network presence.
 First, an outer bound on the secure capacity region is derived for arbitrary network topologies and general number of destinations.
Then, secure transmission schemes are designed and analyzed in terms of achieved rate performance. 
In particular, for the case of two destinations, it is shown that the designed scheme matches the outer bound, hence characterizing the secure capacity region.
It is also numerically verified that the designed scheme matches the outer bound for a special class of networks with general number of destinations, referred to as combination network. 
Finally, for an arbitrary network topology with general number of destinations, a two-phase polynomial time in the network size scheme is designed and its rate performance {is} compared with the capacity-achieving scheme for networks with two destinations.

%% file: introductionMC.tex
Secure network coding~\cite{cai2002secure} considers {the} communication from a source to a number of destinations in the presence of a passive {external} adversary, {with unbounded computational capabilities, but limited network presence.} 
The authors in~\cite{cai2002secure} showed that the source can securely multicast to {all destinations} at a rate of $M-k$, where $M$ is the min-cut capacity between the source and each destination, and $k$ is the number of edges eavesdropped by the adversary. 
In {such a multicast scenario,} all destinations are interested in {receiving the same message.} 

In this paper, we focus on multiple unicast traffic, where a source wishes to securely communicate to a number of destinations, each interested in an independent message.
{Our primal objective lies in characterizing the secure capacity region, by means of derivation of novel outer bounds as well as transmission schemes.}

\subsection{Related Work}
Network coding {was pioneered by} the seminal {work} of Ahlswede et al.~\cite{AhlswedeIT2000}. The authors proved that, if $M$ is the min-cut capacity from the source to each destination, then the source can multicast at a rate $M$ to all the destinations. 
{This result implies that, even if a single destination with min-cut capacity $M$ has access to the entire network resources, this}
destination can only receive {at most at a rate equal to} $M$. 
{Moreover, this result shows that multiple destinations sharing some of the} network resources, can still receive at a rate $M$ if they are interested in the exact same information.
Later, Li et al.~\cite{li2003linear} {proved} that it suffices to use random linear coding operations to {characterize} the multicast capacity. 
Jaggi et al.~\cite{jaggi2005polynomial} designed polynomial time deterministic algorithms {aimed to achieve the multicast capacity.}  
{While for the case of single unicast and multicast traffic the capacity is well-known, the same is not true for the case of networks where multiple unicast sessions take place simultaneously and share some of the network resources. 
For instance, even though the cut-set bound was proved to be tight for some special cases, such as single source with non-overlapping demands and single source with non-overlapping demands and a multicast demand~\cite{KoetterTON2003}, in general it is not tight~\cite{kamath2011generalized}.
It was also recently showed by Kamath et al.~\cite{kamath2014two} that characterizing the capacity of a general network where two unicast sessions take place simultaneously is as hard as characterizing the capacity of a network with general number of unicast sessions.
For the case of single source and two destinations with a non-overlapping demand and a multicast demand, Ramamoorthy et. al~\cite{ramamoorthy2009single} proposed a nice graph theory based approach to characterize the capacity region.}

Information theoretic {security, pioneered by Shannon~\cite{shannon1949communication}, aims at ensuring a reliable and secure communication among trusted parties inside a network such that a passive external eavesdropper does not learn anything about the content of the information exchanged.}
For {point-to-point channels, information theoretic security} can be achieved provided that the communicating {trusted} parties have a pre-shared {key} of entropy {at least equal to the length of the message~\cite{shannon1949communication}.} 
Wyner~\cite{Wyner} showed that, if {the} adversary's channel is a degraded version of the channel to the legitimate destination, {then an} information theoretic secure communication {can be guaranteed} even without the pre-shared keys. 
Moreover, if public feedback is available, {Czap et. al.~\cite{czapP2P} showed that secure communication can be ensured over erasure networks} even when the adversary has {a channel of better quality} than the legitimate receiver. 
{In~\cite{cai2002secure}, Cai et al. characterized the information theoretic secure capacity of a noiseless network with unit capacity edges and with multicast traffic. 
In this work, which was followed by several others~\cite{feldman2004capacity, el2007wiretap}, a source wishes}
to multicast {the same information} to a number of destinations in the presence of a passive {external} adversary eavesdropping any $k$ edges of her choice. 
{ In~\cite{cui2013}, Cui et al.} studied networks {with non-uniform} edge capacities {when} the adversary {is} allowed to {eavesdrop} only {some specific} subsets of edges. 
{Over the past few years, others notions of information theoretic security have been analyzed, such as the case of}
weak information theoretic {security~\cite{bhattad2005weakly,silva2009universal,wei2010}.} 
{Moreover, several different scenarios have been studied, that include: (i) the case of an active adversary, who can indeed corrupt the communication rather than just passively eavesdropping it~\cite{jaggiByzantine,ho2004byzantine,kosutByzantine}; (ii) erasure networks where a public feedback is available~\cite{papadopoulos2015lp,Czaptriangle,CzapLine}; (iii) wireless networks~\cite{MillsWireless2008,dong2009secure}.}

\subsection{Contributions}

{In this paper, we study the problem of characterizing the secure capacity region in a wireline noiseless multiple unicast scenario with uniform edge capacities. In particular, we focus on networks where a source wishes to securely communicate to a number of destinations, each interested in a different message. Our main contributions can be summarized as follows:}
\begin{enumerate}
\item {We derive} an outer bound on the secure capacity region { for} networks with arbitrary topology and arbitrary number of destinations.
{Similar to the multicast scenario~\cite{cai2002secure},} this outer bound depends on the number of edges {that the adversary eavesdrops and on} the min-cut capacities between the source and different subsets of destinations.
\item {We characterize the secure capacity region for networks with arbitrary topology and with two destinations. Towards this end, we design a secure transmission scheme whose achieved rate region is proved to match the derived outer bound. 
In particular, we leverage a key property, referred to as {\it separability}~\cite{ramamoorthy2009single}, in order to select the parts of the network over which: (i) common keys should be multicast, and (ii) encrypted private messages should be communicated.
Our analysis shows that} coding across different unicast sessions helps in {characterizing} the secure capacity even {in scenarios where coding} was not required in the absence of an adversary.
\item { We design} a secure {transmission} scheme for combination networks {with a two-layer topology and arbitrary number of destinations. A key feature of such networks is that they} satisfy the separability property over graphs. 
{In particular, through extensive numerical evaluations, we observed that the designed scheme achieves a secure rate region that matches our derived outer bound, hence suggesting that the proposed scheme {could be} capacity achieving.}
\item {We design a secure transmission scheme} for networks with arbitrary topology and arbitrary number of destinations. This scheme is sub-optimal, but has a polynomial time complexity {in the number of edges and nodes in the network.} 
{In particular,} our scheme works in two phases: in the first phase, we multicast keys using the entire network resources, and in the second phase we communicate encrypted private message packets {using again} the entire network resources. {For the case of two destinations,} we also compare {the secure rate region achieved by this two-phase with the secure capacity region.}
\item {We draw several observations on the derived secure capacity results. For instance, we show that the secure capacity region for two destinations is non-reversible, which is a key difference with respect to the case when there is no adversary. Specifically, we show that, if we switch the role of the source and destinations and we reverse the directions of the edges, then the new secure capacity region differs from the original one. Moreover, for the case of two destinations, we compare the secure capacity region with the capacity region when the adversary is absent. The goal of this analysis is to quantify the rate loss that is incurred to guarantee security.}
\item {We consider} other instances of multiple unicast traffic, {that include: (i) networks with erasure links, and (ii) noiseless networks with two sources and two destinations. In particular, for some specific network topologies, we derive the secure capacity region. This analysis strengthens our previous observation that coding across different unicast sessions is beneficial to ensure a secure communication, even for cases when it is not required in the absence of an adversary.}
\end{enumerate}

\subsection{Paper Organization}
Section~\ref{sec:setting} formally defines the {setup, that is the multiple unicast wireline noiseless network with single source and arbitrary number of destinations, and formulates the problem.}
Section~\ref{sec:outer_bound} derives an outer bound on the secure capacity region. 
Section~\ref{sec:two_destinations} provides a {capacity-achieving secure transmission} scheme for networks with two destinations and arbitrary {topology.}
Section~\ref{sec:combination_nw} {designs a secure transmission} scheme for {combination networks with a two-layer topology and} arbitrary number of destinations. 
Section~\ref{sec:two_phase_scheme} provides a two-phase achievable scheme for networks with arbitrary number of destinations and arbitrary {topology.}
Finally, Section~\ref{sec:discussions} {draws some observations on the derived results, discusses some properties and analyzes other instances of multiple unicast traffic.}

%% file: setupMC.tex
\label{sec:setting}

{Throughout the paper we adopt the following notation convention.}
	Calligraphic letters indicate sets;
	$\emptyset$ is the empty set and
	$\left | \mathcal{A} \right|$ is the cardinality of $\mathcal{A}$; 
	for two sets $\mathcal{A}_1,\mathcal{A}_2$, $\mathcal{A}_1 \subseteq \mathcal{A}_2$ indicates that $\mathcal{A}_1$ is a subset of $\mathcal{A}_2$, $\mathcal{A}_1 \cup \mathcal{A}_2$ indicates the union of $\mathcal{A}_1$ and $\mathcal{A}_2$, $\mathcal{A}_1 \sqcup \mathcal{A}_2$ indicates the disjoint union of $\mathcal{A}_1$ and $\mathcal{A}_2$, $\mathcal{A}_1 \cap \mathcal{A}_2$ is the intersection of $\mathcal{A}_1$ and $\mathcal{A}_2$ and
	$\mathcal{A}_1 \backslash \mathcal{A}_2$ is the set of elements that belong to $\mathcal{A}_1$ but not to $\mathcal{A}_2$;
	$[n_1 : n_2]$ 
	is the set of integers from $n_1$ to $n_2 \geq n_1$;
	$[n]$ 
	is the set of integers from $1$ to $n \geq 1$;
	$[x]^+ := \max\{0, x\}$ for $x \in \mathbb{R}$;
{for a vector $a$, $a^T$ is its transpose vector;
$\text{dim} (A)$ is the dimension of the subspace $A$;
$\mathbf{0}_{i \times j}$ is the all-zero matrix of dimension $i \times j$;
$\mathbf{I}_j$ is the identity matrix of dimension $j$.}

\smallskip

We represent a {wireline noiseless} network with a directed acyclic graph $\mathcal{G} = (\mathcal{V}, \mathcal{E})$, where $\mathcal{V}$ is the set of nodes and $\mathcal{E}$ is the set of directed edges.  
{The} edges represent orthogonal communication links, which are interference-free. 
{In particular,} these links are discrete noiseless memoryless channels over a common alphabet $\mathbb{F}_q$, {i.e., they are of unit capacity} over a $q$-ary alphabet.
If an edge $e \in \mathcal{E}$ connects a node $i$ to a node $j$, we refer to node $i$ as the tail and to node $j$ as the head of $e$, i.e., $\text{tail}(e) = i$ and $\text{head}(e) = j$. 
For each node $v \in \mathcal{V}$, we define $\mathcal{I}(v)$ as the set of all incoming edges of node $v$ and $\mathcal{O}(v)$ as the set of all outgoing edges of node $v$.

{In this network, there is one source node $S$ and $m$ destination nodes $D_i, i \in [m]$. The} 
source node does not have any incoming edges, i.e., $\mathcal{I}(S)=\emptyset$, and {each destination node does} not have any outgoing edges, i.e., {$\mathcal{O}(D_i) = \emptyset, \forall i \in [1:m]$. 
Source $S$} has a message $W_i$ for destination $D_i, i \in [1:m]$. These $m$ messages are assumed to be independent.
Thus, the network consists of multiple unicast traffic, where $m$ unicast sessions take place simultaneously and share the network resources.
{In particular,} each message $W_i, i \in {[m]},$ is of $q$-ary entropy rate $R_i$.
A passive eavesdropper Eve is also present in the network and can wiretap any $k$ edges of her {choice.} We highlight that Eve is an external eavesdropper, i.e., it is not one of the destinations. 

The symbol transmitted over $n$ channel uses on edge $e \in \mathcal{E}$ is denoted as $X_e^n$. In addition, for $\mathcal{E}_t \subseteq \mathcal{E}$ we define $X_{\mathcal{E}_t}^n= \{ X_e^n: e \in \mathcal{E}_t\}$. We assume that the source node $S$ {has} infinite sources of randomness $\Theta$, while the other nodes in the network do not have any randomness. 

 Over {this} network, we are interested in finding all possible feasible $m$-tuples $(R_1,R_2,\ldots, R_m)$ such that each destination $D_i, i \in {[m],}$ reliably decodes the message $W_i$ (with zero error) and Eve {receives no information} about the {content of the} messages.  
In particular, we are interested in {ensuring} perfect information theoretic secure communication, {and hence we aim at characterizing} the secure capacity region, {which is next formally defined.}

\begin{defin} [Secure Capacity Region]
\label{defin:sec_cap}
	A rate $m$-tuple $(R_1, R_2, \ldots, R_m)$ is said to be securely achievable if there exist a block length $n$ and a set of encoding functions  $f_e, \  \forall e \in \mathcal{E}$, with
	\begin{align*}
	X^n_e = 
	\left \{
	\begin{array}{ll}
	f_e \left( W_{{[m]}}, \theta \right) & \mbox{if} \ \mbox{tail}(e)=  S,
	\\
	f_e \left( \{X^n_\ell : \ell \in \mathcal{I}(tail(e))\}  \right) & \mbox{otherwise},
	\end{array} \enspace 
	\right.
	\end{align*}
	{such that each} destination $D_i$ can reliably decode the message $W_i$ i.e., $ H\left( W_i | \{X^n_e: e \in \mathcal{I}(D_i)\}\right) = 0, \ \forall i \in [m]$.
	Moreover, $\forall \ \mathcal{E}_{\mathcal{Z}} \subseteq \mathcal{E}$, $|\mathcal{E}_{\mathcal{Z}}| \leq k$,
	$I \left(W_{{[m]}} ; X^n_{\mathcal{E}_{\mathcal{Z}}}\right)  = 0$ (perfect secrecy requirement). 
{The \textbf{secure capacity region} is}
	the closure of all such feasible rate $m$-tuples. 
\end{defin}

\begin{defin} [Min-cut]
\label{def:MinCut}
{A \textbf{cut} is} an edge set $\mathcal{E}_\mathcal{A} \subseteq \mathcal{E}$, which separates the source $S$ {from} a set of destinations $D_\mathcal{A}:=\{D_i,\ i \in \mathcal{A} \}$. {In a network with unit capacity edges},
{the minimum cut or \textbf{min-cut} is a cut that has the} minimum number of edges.
	\end{defin}


%% file: outer_boundMC.tex
{In this section, we derive an outer bound on the secure capacity region of a multiple unicast wireline noiseless network with a single source and $m$ destinations. In particular, as stated in Theorem~\ref{thm:SecureOB}, this region depends on the min-cut capacities between the source and different subsets of destinations, and on the number of edges that the adversary eavesdrops. The next theorem provides the outer bound region.}

\begin{thm}
	\label{thm:SecureOB}
	An outer bound on the secure capacity region for the multiple unicast traffic over networks with a single source and $m$ destinations is given by
	\begin{align}
	\label{eq:SecureOB}
	R_{\mathcal{A}} \leq [M_{\mathcal{A}} - k]^+,  \ \ \forall \mathcal{A} \subseteq {[m]} \enspace,
	\end{align} 
	where $R_{\mathcal{A}} := \sum\limits_{i \in \mathcal{A}} R_i $ and $M_{\mathcal{A}}$ is the min-cut capacity between the source $S$ and the set of destinations $D_{\mathcal{A}}:= \{D_i: i \in \mathcal{A}\}$.
\end{thm}
	
	\begin{proof}
		Let $\mathcal{E}_{\mathcal{A}}$ be a min-cut between the source $S$ and $D_{\mathcal{A}}$ and $\mathcal{E}_{\mathcal{Z}} \subseteq \mathcal{E}_{\mathcal{A}}$ be the set of $k$ edges wiretapped by Eve, and define $\mathcal{I}(D_{\mathcal{A}}) : = \bigcup_{i \in \mathcal{A}} \mathcal{I}(D_i)$. 
		If $|\mathcal{E}_{\mathcal{A}}| < k$, let $\mathcal{E}_{\mathcal{Z}} = \mathcal{E}_{\mathcal{A}}$. We have,
		\begin{align*}
		n R_{\mathcal{A}}   
		= H(W_{\mathcal{A}})
		& \stackrel{{\rm{(a)}}}{=}   H(W_{\mathcal{A}}) -  H(W_{\mathcal{A}}|X^n_{\mathcal{I}(D_{\mathcal{A}})})  \\
		& \stackrel{{\rm{(b)}}}{=} H(W_{\mathcal{A}}) - H(W_{\mathcal{A}}|X^n_{\mathcal{E}_{\mathcal{A}}}) \\
		& \stackrel{{\rm{(c)}}}{=} I(W_{\mathcal{A}}; X^n_{\mathcal{E}_{\mathcal{Z}}} , X^n_{\mathcal{E}_{\mathcal{A}} \setminus \mathcal{E}_{\mathcal{Z}}})  \\
		& = I(W_{\mathcal{A}}; X^n_{\mathcal{E}_{\mathcal{Z}}}) + I(W_{\mathcal{A}}; X^n_{\mathcal{E}_{\mathcal{A}} \setminus \mathcal{E}_{\mathcal{Z}}} | X^n_{\mathcal{E}_{\mathcal{Z}}} )  \\
		& \stackrel{{\rm{(d)}}}{=} I(W_{\mathcal{A}}; X^n_{\mathcal{E}_{\mathcal{A}} \setminus \mathcal{E}_{\mathcal{Z}}} | X^n_{\mathcal{E}_{\mathcal{Z}}} )  \\
		& \stackrel{{\rm{(e)}}}{\leq} H(X^n_{\mathcal{E}_{\mathcal{A}} \setminus \mathcal{E}_{\mathcal{Z}}}) \\
		& \stackrel{{\rm{(f)}}}{\leq} n [M_{\mathcal{A}} -  k]^+  \enspace,
		\end{align*}
		where $W_{\mathcal{A}}=\{W_i,i \in \mathcal{A}\}$ and where: 
		(i)  the equality in $\rm{(a)}$ follows because of the decodability constraint {(see Definition~\ref{defin:sec_cap})};
		(ii) the equality in $\rm{(b)}$ follows because $X^n_{\mathcal{I}(D_{\mathcal{A}})}$ is a deterministic function of $X^n_{\mathcal{E}_{\mathcal{A}}}$; (iii) the equality in $\rm{(c)}$ follows from the definition of mutual information and since $\mathcal{E}_{\mathcal{A}}=\mathcal{E}_{\mathcal{Z}} \cup \mathcal{E}_{\mathcal{A}\setminus \mathcal{Z}}$;
		(iv) the equality in $\rm{(d)}$ follows because of the perfect secrecy requirement {(see Definition~\ref{defin:sec_cap})};
		(v) the inequality in $\rm{(e)}$ follows since the  entropy of a discrete random variable is a non-negative quantity and because of the `conditioning reduces the entropy' principle;
		(vi) finally, the inequality in $\rm{(f)}$ follows since each link is of unit capacity and since $|\mathcal{E}_{\mathcal{A}} \setminus \mathcal{E}_{\mathcal{Z}}| = [M_{\mathcal{A}} - k]^+$.
		By dividing both sides of the above inequality by $n$ we obtain that $R_{\mathcal{A}}$ in~\eqref{eq:SecureOB} is an outer bound on the {secure} capacity region of the multiple unicast traffic over networks with single source and $m$ destinations.
		This concludes the proof of Theorem~\ref{thm:SecureOB}.
	\end{proof}

	\begin{rem}
		Since the eavesdropper Eve wiretaps any $k$ edges of her choice, intuitively Theorem~\ref{thm:SecureOB} states that, if she wiretaps $k$ edges of a cut with capacity $M$, we can at most hope to reliably transmit at rate $M-k$. 
		However, this holds only for the case of single source; indeed, as we will see in Section~\ref{sec:non_reversibility} through an example, higher rates can be achieved for networks having a single destination and multiple sources.
	\end{rem}

%% file: achievability_two_destinationsMC.tex
\label{sec:Ach2Dest}

{In this section, we} prove that the outer bound in Theorem~\ref{thm:SecureOB} is indeed tight for the case { of $m=2$ destinations. Towards this end, we design a secure transmission scheme whose achievable rate region matches the outer bound in Theorem~\ref{thm:SecureOB}.} In particular, our main result is stated in the following theorem.

\begin{thm}
	\label{thm:SecureLB2}
	The outer bound in~\eqref{eq:SecureOB} is tight for the case $m=2$, i.e., the secure capacity region of the multiple unicast traffic over networks with single source and $m=2$ destinations {is}
	\begin{subequations}
		\label{eq:Ach2}
		\begin{align}
		R_1 & \leq [M_{\{1\}}-k]^+\enspace,
		\\ R_2 & \leq [M_{\{2\}}-k]^+\enspace,
		\\ R_1 + R_2 & \leq [M_{\{1,2\}}-k]^+\enspace. \label{eq:sumrate}
		\end{align}
	\end{subequations}
\end{thm}

\begin{proof}
Clearly, from the result in Theorem~\ref{thm:SecureOB}, the rate region in~\eqref{eq:Ach2} is an outer bound on the {secure} capacity region.
Hence, we now need to prove that the rate region in~\eqref{eq:Ach2} is also achievable.
Towards this end, we start by providing the following definition of \textit{separable} graphs, {which we will leverage in the design of our scheme.}

\begin{defin}[{Separable Graph}]
	\label{def:sepNet}
	A graph $\mathcal{G} = (\mathcal{V}, \mathcal{E})$ with a single source and $m$ destinations is said to be \textbf{separable} if its edge set $\mathcal{E}$ can be partitioned as $\mathcal{E} = \sqcup_{\ell=1}^{2^m-1} \mathcal{E}_{\ell}^\prime$ such that $\mathcal{G}_{\ell}^\prime = (\mathcal{V},\mathcal{E}_{\ell}^\prime), \ {\forall \ell \in [2^m-1]}$ and
	\begin{align}
	M_{\mathcal{A}} & = \sum\limits_{\substack{\mathcal{J}\subseteq {[m]} \\ \mathcal{J} \cap \mathcal{A} \neq \emptyset} } M_{\mathcal{J}}^\star, \ \forall \mathcal{A} \subseteq {[m]} \label{eq:star_mincuts} \enspace,
	\end{align}
	where $M_{\mathcal{A}}$ is the min-cut capacity between the source $S$ and the set of destinations $D_{\mathcal{A}}:= \{D_i: i \in \mathcal{A}\}$  in $\mathcal{G}$ and $M_{\mathcal{J}}^\star$ is the min-cut capacity between the source $S$ and the set of destinations $D_{\mathcal{B}} := \{D_b: b \in \mathcal{B}\}, \ \forall \mathcal{B} \subseteq \mathcal{J}$ for the graph $\mathcal{G}^\prime_{\ell}$ with $\ell \in [1:2^m-1]$ being the decimal representation of the binary vector of length $m$ that has a one in all the positions indexed by $j \in \mathcal{J}$ and zero otherwise, with the least significant bit in the first position. 
\end{defin}

To better understand the above definition, consider a graph $\mathcal{G}$ with $m=2$ destinations. 
Then, the graph $\mathcal{G}$ is separable if it can be partitioned into $3$ graphs such that:
\begin{itemize}
	\item $\mathcal{G}_1^\prime$ has the following min-cut capacities: $M_{\{1\}}^\star$ from $S$ to $D_1$ and zero from $S$ to $D_2$,
	\item $\mathcal{G}_2^\prime$ has the following min-cut capacities: zero from $S$ to $D_1$ and $M_{\{2\}}^\star$ from $S$ to $D_2$,
	\item $\mathcal{G}_3^\prime$ has the following min-cut capacities: $M_{\{1,2\}}^\star$ from $S$ to $D_1$,  $M_{\{1,2\}}^\star$ from $S$ to $D_2$ and $M_{\{1,2\}}^\star$ from $S$ to $\{D_1,D_2\}$,
\end{itemize}
where the quantities $M_{\{1\}}^\star$, $M_{\{2\}}^\star$ and $M_{\{1,2\}}^\star$ can be computed using the following set of equations:
\begin{subequations}
	\begin{align}
	M_{\{1\}} &= M_{\{1\}}^\star + M_{\{1,2\}}^\star \enspace, \\
	M_{\{2\}} &= M_{\{2\}}^\star + M_{\{1,2\}}^\star \enspace, \\
	M_{\{1,2\}} & = M_{\{1\}}^\star +M_{\{2\}}^\star +M_{\{1,2\}}^\star\enspace. \label{eq:M12star}
	\end{align}
\label{eq:Mstar}
\end{subequations}
We now state the following lemma, which is a consequence of~\cite[Theorem 1]{ramamoorthy2009single} and we will use to prove the achievability of the rate region in~\eqref{eq:Ach2}.
	
	\begin{lem}
		\label{lemma:sep}
		Any graph with a single source and $m=2$ destinations is separable.
	\end{lem}

	For completeness we report the proof of Lemma~\ref{lemma:sep} in Appendix~\ref{app:seperable_graph_m_2}.
	By leveraging the result in Lemma~\ref{lemma:sep}, we are now ready to prove Theorem~\ref{thm:SecureLB2}. 
	In particular, we consider two cases depending on the value of $k$ (i.e., the number of edges {that} the eavesdropper wiretaps).
	Without loss of generality, we assume that $k < \min_{i \in {[2]}} M_i$, as otherwise secure communication to the set of destinations $\{D_i : k \geq M_i\}$ is not possible at any rate, and hence we can just remove this set of destinations from the network.
	
	\begin{enumerate}
		\item {\bf{Case 1:}} $k \geq M_{\{1,2\}}^\star$.
		In this case, by substituting the quantities in~\eqref{eq:Mstar} into~\eqref{eq:Ach2}, we obtain that the constraint in~\eqref{eq:sumrate} is redundant.
		Thus, we will now prove that the rate pair $(R_1, R_2) = (M_{\{1\}} - k, M_{\{2\}} -k)$ is securely achievable, which along with the time-sharing argument proves the achievability of the entire {rate} region in~\eqref{eq:Ach2}.
		
		We denote with $K_1, K_2, \ldots, K_k$ the $k$ key packets and with $W_i^{(1)}, W_i^{(2)}, \ldots, W_i^{(R_i)}$ (with $i\in {[2]}$) the $R_i$ message packets for $D_i$.
		With this, our scheme is as follows:
		\begin{itemize}
			\item We multicast $K_i, \forall i \in {[M_{\{1,2\}}^\star]}$, to both $D_1$ and $D_2$ using $\mathcal{G}_3^\prime$, which has edges denoted by $\mathcal{E}_3^\prime$. 
			This is possible {since} $\mathcal{G}_3^\prime$ has a min-cut capacity $M_{\{1,2\}}^\star$ to both $D_1$ and $D_2$ (see Definition~\ref{def:sepNet}). 
			
			\item 
			We unicast $K_{\ell}, \forall \ell \in [M_{\{1,2\}}^\star+1:k]$, to $D_i, \forall i \in {[2]}$, 
			using $k - M_{\{1,2\}}^\star$ paths out of the $M_{\{i\}}^\star$ disjoint paths in $\mathcal{G}_i^\prime$.
			We denote by $\hat{\mathcal{E}}_i$ the set that contains all the first edges of these paths. 
			Clearly, $|\hat{\mathcal{E}}_i| = k - M_{\{1,2\}}^\star, \forall i \in {[2]}$.
			Notice that $\hat{\mathcal{E}}_i \subseteq \mathcal{E}_i^\prime, \forall i \in {[2]}$ (see Definition~\ref{def:sepNet}). 
			\item 
			We send the $R_i, \forall i \in {[2],}$ encrypted message packets {(i.e., encoded with the keys)} of $D_i$ on the remaining $M_{\{i\}}^\star - k + M_{\{1,2\}}^\star$ disjoint paths in $\mathcal{G}_i^\prime$.
			We denote by $\bar{\mathcal{E}}_i$ the set that contains all the first edges of these paths in $\mathcal{G}_i^\prime$. Clearly, $|\bar{\mathcal{E}}_i| = R_i, \forall i \in {[2]}$, $\bar{\mathcal{E}}_i \subseteq \mathcal{E}_i^\prime$ and $\bar{\mathcal{E}}_i \cap \hat{\mathcal{E}}_i = \emptyset $ (see Definition~\ref{def:sepNet}).
		\end{itemize}
		
		This scheme achieves $R_i = M_{\{i\}}^\star - k + M_{\{1,2\}}^\star = M_{\{i\}} -k, \forall i \in [1:2]$, where the second equality follows by using the definitions in~\eqref{eq:Mstar}.
		Now we prove that this scheme is also secure. 
		We start by noticing that, thanks to Definition~\ref{def:sepNet}, the edges $\mathcal{E}_3^\prime$, $\hat{\mathcal{E}}_i$ and $\bar{\mathcal{E}}_i$, with $i \in {[2]}$, do not overlap.
		We write these transmissions in a matrix form (with $G$ and $U$ being {the} encoding matrices) and we obtain
		
		\begin{align*}
		\left[ \begin{array}{c} X_{\mathcal{E}_3^\prime} \\ X_{\hat{\mathcal{E}}_1} \\ X_{\hat{\mathcal{E}}_2} \end{array}\right] & = \underbrace{\left[ \begin{array}{cccc} g_{11} & g_{12} & \ldots & g_{1k} \\ g_{21} & g_{22} & \ldots & g_{1k} \\ \vdots & \vdots & \ddots & \vdots \\ g_{\ell 1} & g_{\ell2} & \ldots & g_{\ell k} \\  \end{array} \right]}_G \left[ \begin{array}{c} K_1 \\ K_2 \\ \vdots \\ K_k \end{array}\right], \ \ell = |\mathcal{E}_3^\prime| + 2 \left (k-M_{\{1,2\}}^\star \right )\enspace,
		\\
		\left[ \begin{array}{c} X_{\bar{\mathcal{E}}_1} \\ X_{\bar{\mathcal{E}}_2} \end{array}\right] & = \underbrace{\left[ \begin{array}{cccc} u_{11} & u_{12} & \ldots & u_{1k} \\ u_{21} & u_{22} & \ldots & u_{2k} \\ \vdots & \vdots & \ddots & \vdots \\ u_{r1} & u_{r2} & \ldots & u_{rk} \\  \end{array} \right]}_U  \left[ \begin{array}{c} K_1 \\ K_2 \\ \vdots \\ K_k \end{array}\right]  \oplus \left[ \begin{array}{c} W_1^{(1)} \\ \vdots \\ W_1^{(R_1)} \\ W_2^{(1)} \\ \vdots \\ W_2^{(R_2)} \end{array}\right],  \ r = R_1 + R_2\enspace.
		\end{align*}
		The eavesdropper Eve wiretaps $k_1 \leq k$ edges from the collection of edges $\{\mathcal{E}_3^\prime, \hat{\mathcal{E}}_1,\hat{\mathcal{E}}_2\}$, over which the linear combinations $X_{\mathcal{E}_3^\prime}$, $X_{\hat{\mathcal{E}}_1}$ and $X_{\hat{\mathcal{E}}_2}$ of keys are transmitted, and $k_2 = k - k_1$ edges from the collection of edges $\{\bar{\mathcal{E}}_1,\bar{\mathcal{E}}_2\}$ over which the messages {encoded} with the keys $X_{\bar{\mathcal{E}}_1}$ and $X_{\bar{\mathcal{E}}_2}$ are transmitted. 
		We here note that on the other edges $\mathcal{E}\backslash \{\mathcal{E}_3^\prime \cup \hat{\mathcal{E}}_1 \cup \bar{\mathcal{E}}_1 \cup \hat{\mathcal{E}}_2 \cup \bar{\mathcal{E}}_2 \}$ of the network, we either do not transmit any symbol or simply route the symbols from $\{ X_{\bar{\mathcal{E}}_1}, X_{\bar{\mathcal{E}}_2}, X_{\hat{\mathcal{E}}_1}, X_{\hat{\mathcal{E}}_2} \}$ (corresponding to the symbols transmitted on disjoint paths). Thus, without loss of generality, we can assume that Eve does not wiretap any of these edges. 
		Since the first $|\mathcal{E}_3^\prime|$ rows of $G$ (i.e., those that correspond to multicasting {the keys)} are determined by the network coding scheme for multicasting~\cite{AhlswedeIT2000}, we assume that we do not have any control over the construction of $G$. 
		
		Thus, we would like to construct the code matrix $U$ such that all the linear combinations of the keys used to encrypt the messages on $k_2$ edges are mutually independent and are independent from the linear combinations of the keys wiretapped on the $k_1$ edges (notice that this makes the symbols wiretapped by the eavesdropper completely independent from the messages).
		In particular, 
		since in the worst case Eve wiretaps $k_1$ edges which are independent linear combinations, we would like that any matrix formed by $k_1$ independent rows of the matrix $G$ and $k_2$ rows of the matrix $U$ is full rank. 
		Since there is a finite number of such choices and the determinant of each of these possible matrices can be written in a polynomial form -- which is not identically zero -- as a function of the entries of $U$, then we can choose the entries of $U$ such that all these matrices are invertible.
		Thus, we can always construct the code matrix $U$ such that the edges wiretapped by Eve have independent keys and hence Eve does not get any information about the message packets, i.e., the scheme is secure.
		This implies that the rate pair $(R_1, R_2) = (M_{\{1\}} - k, M_{\{2\}} -k)$ is securely 
		achievable.
		
		\item {\bf{Case 2:}} $k < M_{\{1,2\}}^\star$. 
		By substituting the quantities in~\eqref{eq:Mstar}, the rate region in~\eqref{eq:Ach2} becomes
		\begin{subequations}
			\label{eq:Ach2_star}
			\begin{align}
			R_i & \leq M_{\{i\}} -k = M_{\{i\}}^\star + M_{\{1,2\}}^\star - k, \forall i \in {[2]}\enspace,\\
			R_1 + R_2 & \leq M_{\{1,2\}} -k = M_{\{1\}}^\star + M_{\{2\}}^\star + M_{\{1,2\}}^\star - k\enspace.
			\end{align}
		\end{subequations} 
		
		We now show that we can achieve the following two corner points
		i.e., the rate pair
		\begin{align}
		(R_1, R_2) & = \left( \alpha (M_{\{1,2\}} - M_{\{2\}}) +(1-\alpha) (M_{\{1\}} - k), \right.  \nonumber
		\\& \quad \  \left. 
		\alpha (M_{\{2\}} - k)+(1-\alpha) (M_{\{1,2\}} - M_{\{1\}}) \right) \nonumber
		\\& \stackrel{{\rm{(a)}}}{=} (M_{\{1\}}^\star + \alpha (M_{\{1,2\}}^\star - k), M_{\{2\}}^\star + (1-\alpha) (M_{\{1,2\}}^\star - k))\enspace,
		\label{eq:RatePairCase2}
		\end{align}
		for $\alpha \in \{0,1\}$, where the equality in $\rm{(a)}$ follows by using the definitions in~\eqref{eq:Mstar}. This along with the time-sharing argument proves the achievability of the entire {rate} region in~\eqref{eq:Ach2_star}.
		We recall that we denote with $K_1, K_2, \ldots, K_k$ the $k$ key packets and with $W_i^{(1)}, W_i^{(2)}, \ldots, W_i^{(R_i)}$ (with $i\in {[2]}$) the $R_i$ message packets for $D_i$.
		With this, our scheme is as follows:
		\begin{itemize}	
			\item Using the graph $\mathcal{G}_3^\prime$ we multicast to both destinations $D_1$ and $D_2$: 
			(i) $K_i, \forall i \in {[k]}$, 
			(ii) $\alpha (M_{\{1,2\}}^\star - k)$ encrypted message packets {(i.e., encoded with the keys)} for $D_1$ and 
			(iii) $(1-\alpha) (M_{\{1,2\}}^\star - k)$ encrypted message packets for $D_2$. 
			Recall that the edges of the graph $\mathcal{G}_3^\prime$ are denoted by $\mathcal{E}_3^\prime$ (see Definition~\ref{def:sepNet}).
			We also highlight that the message packets multicast to the two destinations are encrypted using the key packets, where the encryption is based on the secure network coding result on multicasting~\cite{cai2002secure}, which ensures perfect security from an adversary wiretapping any $k$ edges.
			\item 
			We send  
			$M_{\{i\}}^\star$ encrypted message packets of $D_i$ on the $M_{\{i\}}^\star$ disjoint paths to $D_i$ in the graph $\mathcal{G}_i^\prime$, and denote by $\hat{\mathcal{E}}_i$ the set that contains all the first edges of these paths for $i \in {[2]}$.
		\end{itemize}
		
		This scheme achieves the rate pair in~\eqref{eq:RatePairCase2}.
		Now we prove that this scheme is also secure. 
		For ease of representation, in what follows we let $R_1^\star = \alpha (M_{\{1,2\}}^\star - k)$ and $R_2^\star = (1-\alpha) (M_{\{1,2\}}^\star - k)$.
		We again notice that, thanks to Definition~\ref{def:sepNet}, the edges $\mathcal{E}_3^\prime$, $\hat{\mathcal{E}}_1$ and $\hat{\mathcal{E}}_2$ do not overlap.
		We write these transmissions in a matrix form (with $G$, $U$ and $S$ being encoding matrices) and we obtain,
		\begin{align*}
		& X_{\mathcal{E}_3^\prime}  = \underbrace{\left[ \begin{array}{cccc} g_{11} & g_{12} & \ldots & g_{1k} \\ g_{21} & g_{22} & \ldots & g_{1k} \\ \vdots & \vdots & \ddots & \vdots \\ g_{\ell 1} & g_{\ell2} & \ldots & g_{\ell k}  \end{array} \right]}_G \left[ \begin{array}{c} K_1 \\ K_2 \\ \vdots \\ K_k  \end{array}\right]  \oplus \underbrace{\left[ \begin{array}{cccc} s_{11} & s_{12} & \ldots & s_{1k} \\ s_{21} & s_{22} & \ldots & s_{1k} \\ \vdots & \vdots & \ddots & \vdots \\ s_{\ell 1} & s_{\ell2} & \ldots & s_{\ell k}  \end{array} \right]}_S \left[ \begin{array}{c} W_1^{(1)} \\ \vdots \\ W_1^{(R_1^\star)} \\ W_2^{(1)} \\ \vdots \\ W_2^{(R_2^\star)}  \end{array}\right] ,  \ \ell = |\mathcal{E}_3^\prime|\enspace,
		\end{align*}
		\begin{align*}
		& \left[ \begin{array}{c} X_{\hat{\mathcal{E}}_1} \\ X_{\hat{\mathcal{E}}_2} \end{array}\right]  \!\!= \!\! \underbrace{\left[ \begin{array}{cccc} u_{11} & u_{12} & \ldots & u_{1k} \\ u_{21} & u_{22} & \ldots & u_{2k} \\ \vdots & \vdots & \ddots & \vdots \\ u_{r1} & u_{r2} & \ldots & u_{rk} \\  \end{array} \right]}_U  \left[ \begin{array}{c} K_1 \\ K_2 \\ \vdots \\ K_k \end{array}\right] \!\! \oplus \!\! \left[ \begin{array}{c} m_1^{(R_1^\star + 1)} \\ \vdots \\ W_1^{(R_1)} \\ W_2^{(R_2^\star + 1)} \\ \vdots \\ W_2^{(R_2)} \end{array}\right]\!, \ r \!=\! R_1 \!+\! R_2 \!-\! (M_{\{1,2\}}^\star \!-\! k)\enspace.
		\end{align*}
		
		The eavesdropper Eve wiretaps $k_1 \leq k$ edges from $\mathcal{E}_3^\prime$, over which the linear combinations $X_{\mathcal{E}_3^\prime}$ of key packets and message packets are sent, and $k_2 = k - k_1$ edges from the collection of edges $\{ \hat{\mathcal{E}}_1,\hat{\mathcal{E}}_2 \}$ over which the messages {encoded} with the keys $X_{\hat{\mathcal{E}}_1}$ and $ X_{\hat{\mathcal{E}}_2}$ are transmitted. 
		Similar to Case $1$, on the other edges $\mathcal{E}\backslash \{\mathcal{E}_3^\prime \cup \hat{\mathcal{E}}_1 \cup \hat{\mathcal{E}}_2\}$ of the network, we either do not transmit any symbol or simply route the symbols from $\{X_{\hat{\mathcal{E}}_1}, X_{\hat{\mathcal{E}}_2} \}$ (corresponding to the symbols transmitted on disjoint paths). Thus, without loss of generality, we can assume that the eavesdropper does not wiretap any of these edges. 
		Since the matrices $G$ and $S$ are determined by the secure network coding scheme for multicasting~\cite{cai2002secure}, we do not have any control over their construction.
		Thus, we would like to construct the code matrix $U$ in order to ensure security. 
		Again, similar to the argument used in Case $1$, we can create $U$ such that any subset of $k_2$ rows of $U$ are linearly independent and not in the span of any subset of $k_1$ rows of $G$. 
		With this, the keys used to encrypt the messages over any $k_2$ edges of $\{\hat{\mathcal{E}}_1, \hat{\mathcal{E}}_2\}$ are mutually independent and independent from the keys used over any $k_1$ edges of $\mathcal{E}_3^\prime$. This, together with the fact that the messages transmitted using $\mathcal{G}_3^\prime$ are already secure, makes our scheme secure. 
		This implies that the rate pair $(R_1,R_2)$ in~\eqref{eq:RatePairCase2} is securely achievable.

	\end{enumerate}
	This concludes the proof of Theorem~\ref{thm:SecureLB2}.
\end{proof}

\begin{figure}[t]
	\centering
	\includegraphics[width=400px]{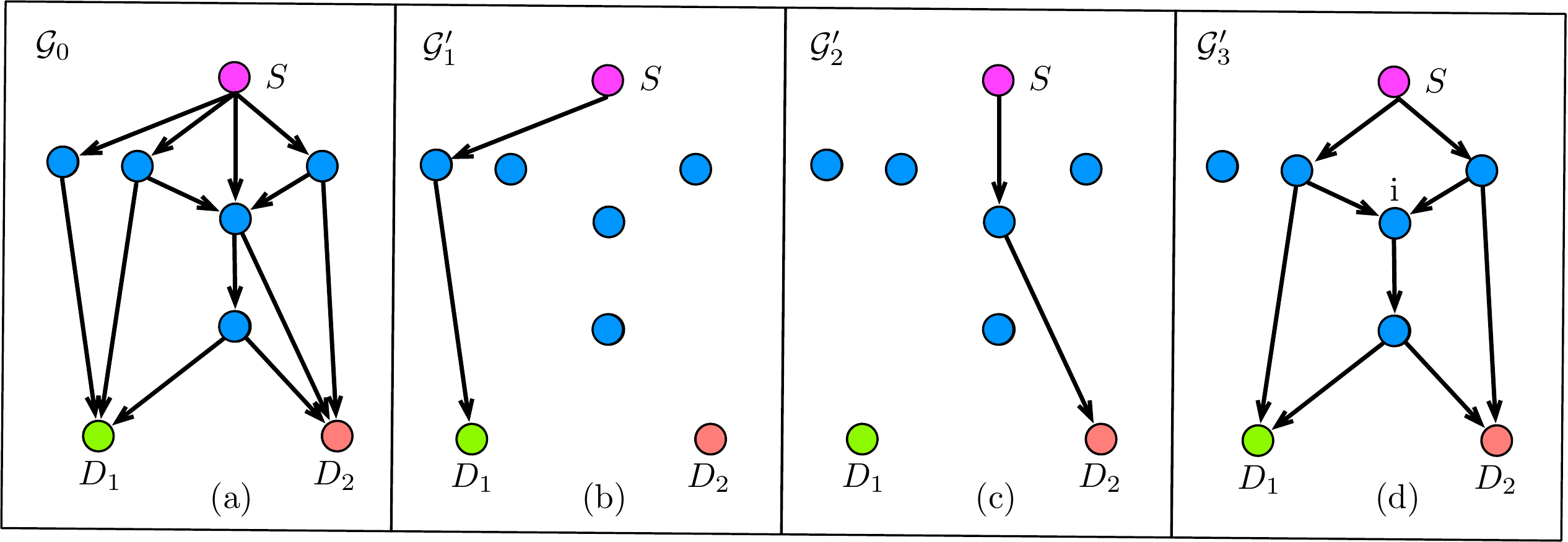}
	\caption{A $2$-destination separable network $\mathcal{G}_0$ in {(a)} and its partition graphs $\mathcal{G}_i^\prime, i \in {[3]}$ {in (b)-(d)}.}
	\label{fig:toy_example_partitions}
\end{figure} 

{ We next illustrate the above described scheme for the network $\mathcal{G}_0$ in Fig.~\ref{fig:toy_example_partitions}(a). We first note that $\mathcal{G}_0$ has min-cut capacities $M_{\{1\}} = M_{\{2\}}  = 3$ and $M_{\{1,2\}} = 4$, and it can be partitioned into three graphs $\mathcal{G}_i^\prime, i \in [3],$ as shown in Figs.~\ref{fig:toy_example_partitions}(b)-(d), with min-cut capacities equal to $M_{\{1\}}^\star = M_{\{2\}}^\star  = 1$ and $M_{\{1,2\}}^\star = 2$, respectively. We assume that the adversary eavesdrops any $k=2$ edges of her choice. For this case, the source should be able to securely communicate at a rate $(R_1,R_2) = (1,1)$ towards the $m=2$ destinations. This rate pair can be achieved as follows:
\begin{enumerate}
\item Over the set of edges in $\mathcal{G}^\prime_1$, the source transmits $W_1 \oplus K_1 \oplus 2 K_2 $; the intermediate node simply routes this transmission to $D_1$;
\item Over the set of edges in $\mathcal{G}^\prime_2$, the source transmits $W_2 \oplus K_1 \oplus 3 K_2 $; the intermediate node simply routes this transmission to $D_2$;
\item Over the set of edges in $\mathcal{G}^\prime_3$, the source transmits $K_1$ to one intermediate node and $K_2$ to the other intermediate node. 
The intermediate node denoted as $\text{i}$ in Fig.~\ref{fig:toy_example_partitions}(d) receives $K_1$ and $K_2$ and transmits $K_1 \oplus K_2$ on its outgoing edges. Note that with this strategy $D_1$ and $D_2$ receives $K_1$ and $K_2$. It therefore follows that $D_i,i \in [2],$ can successfully recover~$W_i$.
\end{enumerate}

We conclude this section with an observation on separable graphs. In particular, we show that, although for the case of $m=2$ destinations any graph is separable, in general the same does not hold for $m \geq 3$.

\begin{rem}
Consider the network in Fig.~\ref{fig:non_sep_nw}, which consists of $m=3$ destinations and has the following min-cut capacities: $M_{\{1\}} = 1$, $M_{\{2\}} = 1$, $M_{\{3\}} = 1$, $M_{\{1,2\}} = 2$, $M_{\{2,3\}} = 2$, $M_{\{1,3\}} = 2$ and $M_{\{1,2,3\}} = 2$. 
With this, we can find $M_{\mathcal{J}}^\star, \  \mathcal{J} \subseteq [3]$, by solving~\eqref{eq:star_mincuts}. In particular, we obtain: $M_{\{1\}}^\star = M_{\{2\}}^\star = M_{\{3\}}^\star = 0$, $M_{\{1,2\}}^\star = M_{\{2,3\}}^\star = M_{\{1,3\}}^\star = 1$ and $M_{\{1,2,3\}}^\star = -1$. Since a graph can not have a negative min-cut capacity, we readily conclude that a separation of the form defined in Definition~\ref{def:sepNet} is not possible.  
\end{rem}
}

\begin{figure}[t]
	\centering
	\includegraphics[width=0.3\columnwidth]{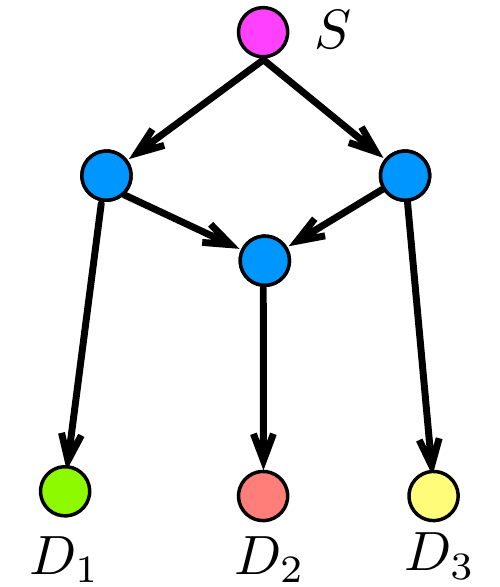}
	\caption{{Example of a non-separable graph.}}
	\label{fig:non_sep_nw}
\end{figure}

%% file: combination_networksMC.tex
{In this section, we focus on a special class of networks, referred to as combination networks, and design a secure transmission scheme. Before delving into the study of such networks, we note that the capacity-achieving scheme for $m=2$ destinations described in Section~\ref{sec:two_destinations} uses some part of the network to multicast the keys and the remaining part to communicate the encrypted messages (i.e., messages encoded with the keys).
Therefore, we now ask} the following question: can we extend this idea to get {a capacity-achieving} scheme for networks with arbitrary number of destinations? In other words, can we separate, {over different parts of the network,} the key transmissions and {the message} transmissions? {We next show that this is not possible through an example.}
\begin{figure}[t]
	\centering
	\includegraphics[width=0.5\columnwidth]{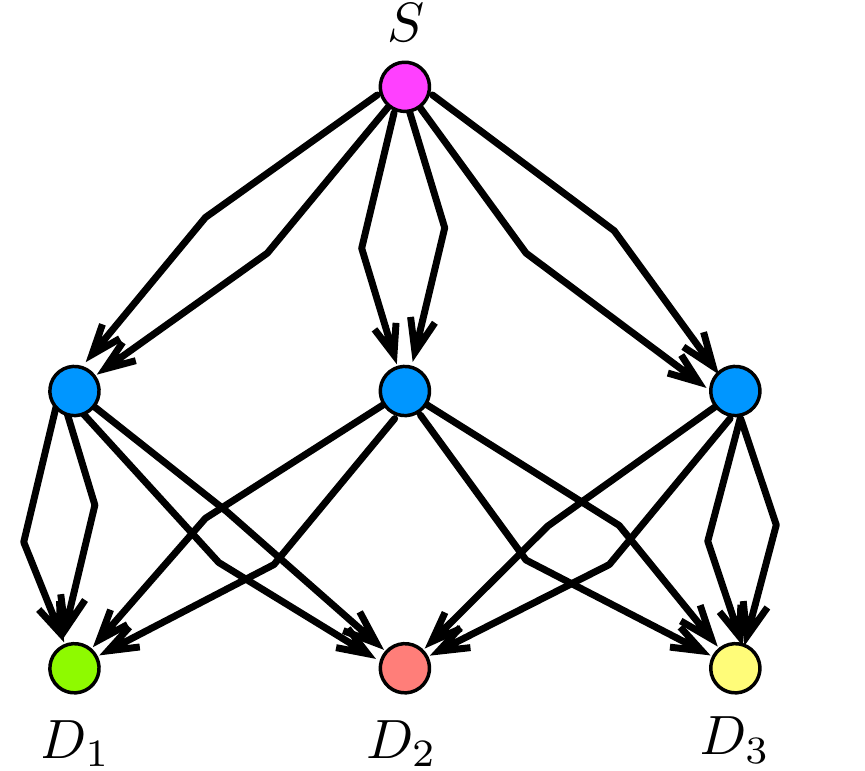}
	\caption{Network example to show that {using different parts of the network to transmit the keys and the encrypted messages is not optimal.}}
	\label{fig:net_key_sprtn}
\end{figure}  
Consider the network shown in Fig.~\ref{fig:net_key_sprtn}, { which consists of $m=3$ destinations, and where the} adversary can eavesdrop any {$k=3$} edges of her choice. 
{For this network we have the following min-cut capacities: $M_{\{1\}} = M_{\{2\}} = M_{\{3\}} =4$, $M_{\{1,2\}} = M_{\{1,3\}} = M_{\{2,3\}} =M_{\{1,2,3\}} =6$. We would like to show that the triple $(R_1,R_2,R_3) = (1,1,1)$ -- obtained from the outer bound in Theorem~\ref{thm:SecureOB} -- can not be achieved when the key packets and the encrypted messages are transmitted over different parts of the network.}
It is not difficult to see that, out of {the $6$} outgoing edges from the source, multicasting $3$ keys\footnote{Note that $3$ keys are required since the adversary eavesdrops $k=3$ edges of her choice.} {requires a number of edges strictly greater than $4$. Thus, we would be left with strictly less than $2$ edges, which are not sufficient to transmit $3$ message packets, i.e., one for each destination.} 
{ It therefore follows that, with this strategy, the rate triple $(R_1,R_2,R_3) = (1,1,1)$} can not be securely achieved. 

However, we now {design a transmission} scheme, where the messages and the keys are encoded jointly and {show that the rate triple $(R_1,R_2,R_3) = (1,1,1)$ can indeed be securely achieved. In what follows, we let: (i) $W_i, i \in[3],$ be the message for $D_i$, (ii) $K_i, i \in[3]$, be the three random packets transmitted by the source to guarantee security (recall that the eavesdropper wiretaps any $k=3$ edges of her choice), and (iii) $X_i, i \in [6]$, be the symbols transmitted by the source on its outgoing edges (enumerated from left to right with reference to Fig.~\ref{fig:net_key_sprtn}). {Intermediate nodes simply route the received symbols on their outgoing edges, i.e., there is no coding operation at the intermediate nodes.} With this, we now define our scheme in matrix form as follows,} 
\begin{align}
\begin{bmatrix}
X_1 \\
X_2 \\
X_3 \\
X_4 \\
X_5 \\
X_6
\end{bmatrix}
=
\underbrace{
	\begin{bmatrix}
	0 & 0 & 0 & 1 & 0 & 0 \\
	0 & 0 & 0 & 1 & 1 & 1 \\
	0 & 0 & 0 & 1 & 2 & 4 \\
	1 & 0 & 0 & 1 & 3 & 2 \\
	4 & 6 & 4 & 1 & 4 & 2 \\
	2 & 4 & 2 & 1 & 5 & 4 \\
	\end{bmatrix}}_{B}
\begin{bmatrix}
W_1 \\
W_2 \\
W_3 \\
K_1 \\
K_2 \\
K_3
\end{bmatrix},
\label{eq:SystModelIO2}
\end{align}
{ where $B \in \mathbb{F}_7$ is the encoding matrix. We start by noting that,} for every set of {$3$} edges, we have linearly independent keys added to the linear combinations of messages, {and hence the scheme is secure. Moreover, the destinations can successfully recover their message by using the following decoding scheme:}    
\begin{itemize}
	\item Destination 1: $W_1  = 6 X_1 + 3 X_2 + 4 X_3 +  X_4$,
	\item Destination 2: $W_2  = 6 X_1 + 4 X_2 + 3 X_5 +  X_6$,
	\item Destination 3: $W_3  = 5 X_3 + 6 X_4 +  X_5 +  2 X_6$ .		
\end{itemize}

{Thus, the rate triple $(R_1,R_2,R_3) = (1,1,1)$ can be securely achieved. This example shows that using different parts of the network to transmit the keys and the encrypted messages, in general is not optimal.} This is partially due to the fact that destinations do not need to decode {each key individually, as long as they can successfully recover their message.}

\subsection{{Secure Transmission Scheme}}

We now {leverage the observations drawn for the network in Fig.~\ref{fig:net_key_sprtn} to design a secure transmission scheme for a class of networks,  referred to as combination networks.
As formally defined in Definition~\ref{defn:combin_networks}, these networks have a two-layer topology, they are separable (see Definition~\ref{def:sepNet}) and they can have an arbitrary number of destinations.}

\begin{defin}[{Combination Network}]
	\label{defn:combin_networks}
A combination network parameterized by $(t,m, \{\mathcal{M}_i, \forall i \in [m]\} )$ is defined as follows. The source node $S$ is connected to $t$ {intermediate} nodes that form the first layer of the network. {Each intermediate node has one incoming edge from the source.} On the second layer, there are $m$ destination nodes $D_1, D_2, \ldots, D_m$, such that {$D_i, i \in [m],$} is connected to a subset of {intermediate nodes given by $\mathcal{M}_i \subseteq [t]$. {Each destination has at most one incoming edge from intermediate node $i \in [t]$.}} 
\end{defin}

{An example of combination network is shown in Fig.~\ref{fig:example_t6_m3}, for which}
$t = 6$, $m =3$, and $\mathcal{M}_1 = \{1,2,4\}$, $\mathcal{M}_2=\{3,4,5,6\}$ and $\mathcal{M}_3 = \{2,3\}$.

\begin{figure}[t]
	\centering
	\includegraphics[width=0.5\linewidth]{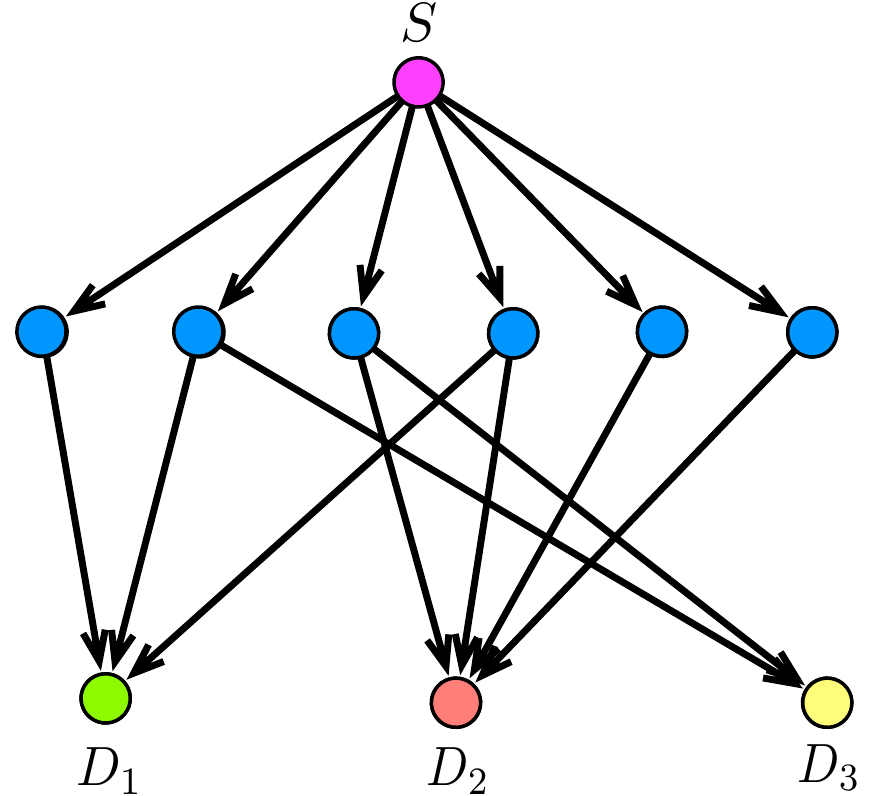}
	\caption{An example of a combination network with $t=6$ and $m=3$.}
	\label{fig:example_t6_m3}
\end{figure}   

The rate region achieved by our {designed secure transmission} scheme depends on $m$ carefully constructed null spaces. {We observe that each receiver $D_i$ will use  $R_i$ vectors to ``decode" (we will refer to these as ``decoding vectors"), to retrieve the $R_i$ messages it requests. {These vectors need to enable to cancel out the keys}, need to be linearly independent, and need to use only the encoded transmissions of the source that the receiver has access to. The intuition behind our scheme design is to construct the null spaces where these decoding vectors reside. 
 }
 We {now show} the construction of {such} null spaces.
Consider a {Vandermonde matrix} $V$ with $k$ rows and $t$ columns as shown in~\eqref{eq:vandermonde}, {where $\alpha_i \in \mathbb{F}_q,  \forall i \in [t]$, are all distinct.}
\begin{align}
\label{eq:vandermonde}
V = \left[ \begin{array}{cccc} 1 & 1 & \ldots & 1 \\ \alpha_1 & \alpha_2 & \ldots & \alpha_t \\ \vdots \\ \alpha_1^{k-1} & \alpha_2^{k-1} & \ldots & \alpha_t^{k-1}
\end{array}\right].
\end{align}
{Note that $t > k$, otherwise secure communication is not possible, i.e., if $k \geq t$, then the adversary wiretaps the entire communication from the source.}
Since $V$ is a Vandermonde matrix, it has the property that any {any $k \times k$ submatrix is full rank, i.e., any} set of $k$ columns are linearly independent. Moreover, the right null space of $V$ {is} of dimension $t-k$.  
{This matrix will be used to encode the keys.}
For each destination $D_i$,  $i \in [m]$, we consider the following matrix $V_i$:
\begin{align}
\label{eq:vandermonde_i}
V_i & = \left[ \begin{array}{c} V \\ C_i \end{array} \right],
\end{align}
where $C_i$ is a matrix having $t - \left| \mathcal{M}_i \right|$  rows and $t$ columns. The rows of $C_i$ are given by $\{c_j,  j \in [t]\setminus \mathcal{M}_i \}$, where $c_j$ {is a vector of length $t$ with a one in the $j$-th position and zeros everywhere else. {The role of the matrix $C_i$ is to restrict receiver $D_i$ to only use the source encoded transmissions it actually has access to.} 
For instance, with reference to the example in Fig.~\ref{fig:example_t6_m3}, we would have
\begin{align*}
C_1 = \begin{bmatrix}
0 & 0 & 1 & 0 & 0 & 0 \\
0 & 0 & 0 & 0 & 1 & 0 \\
0 & 0 & 0 & 0 & 0 & 1
\end{bmatrix},
\quad
C_2 = \begin{bmatrix}
1 & 0 & 0 & 0 & 0 & 0 \\
0 & 1 & 0 & 0 & 0 & 0 
\end{bmatrix},
\quad 
C_3 = \begin{bmatrix}
1 & 0 & 0 & 0 & 0 & 0 \\
0 & 0 & 0 & 1 & 0 & 0 \\
0 & 0 & 0 & 0 & 1 & 0 \\
0 & 0 & 0 & 0 & 0 & 1 
\end{bmatrix}.
\end{align*}
With this construction, we have that: (i)} 
all rows of $V$ are linearly independent ({ because of the property of the Vandermonde matrix in~\eqref{eq:vandermonde}); (ii)} all rows of $C_i$ are linearly independent; {(iii)} any vector in the span of the rows of $V$ {has} a weight {of} at least $t-k+1$ (because $V$ is {the} generator matrix {of a} $(t,k,t-k+1)$ {Maximum Distance Separable (MDS)} code); {(iv)} any vector in the span of the rows of $C_i$ {has a weight of} at most $t - \left| \mathcal{M}_i \right|$. 
Thus, as long as $t-k+1 > t - \left| \mathcal{M}_i \right|$, i.e.,  $\left| \mathcal{M}_i \right| \geq k$, {then} all rows of $V_i$ are linearly independent.
Let $N_i$ be the right null space of $V_i$, {then} $N_i$ will be of dimension $t - (k+ t - \left| \mathcal{M}_i \right|)$, i.e., $\left| \mathcal{M}_i \right| -k$ if $\left| \mathcal{M}_i \right| \geq k$. 
{For the case when} $\left| \mathcal{M}_i \right| < k$, $V_i$ will be a full rank matrix {of rank $t$} and $N_i$ will be {an empty space.} Thus,
\begin{align}
\text{dim} (N_i) & = \left [ \left| \mathcal{M}_i \right| -k \right ]^+, \forall i \in [m]. \label{eq:dim_N_i}
\end{align} 
{With the definition of the null spaces $N_i, i \in [m]$, above, we are now ready to present our achievable rate region for combination networks. }

\begin{prop}
	\label{thm:SecureAch_combi}
For the combination network $(t,m, \{\mathcal{M}_i, \forall i \in [m]\} )$, {assume that, for all $i \in [m]$, we can select $R_i$ vectors from $N_i$ such that the selected $\sum\limits_{i=1}^{m} R_i$ vectors are {linearly} independent. Then, the rate tuple $(R_1, R_2, \ldots, R_m)$ is securely achievable and the convex hull of all these feasible rate tuples is the achievable region.}
\end{prop}

\begin{proof}

{We next describe the different encoding/decoding operations of our scheme for a specific tuple $(R_1, R_2, \ldots, R_m)$ that satisfies the condition in Proposition~\ref{thm:SecureAch_combi}\footnote{We  assume that the $R_i$s with $i\in [m]$ are all integers. This assumption is without loss of generality since: (i) rational $R_i$s can be characterized by time-sharing the network with integer values of achievable rate tuples; (ii) rate tuples over real numbers can be approximated with rational rate tuples.}. Towards this end, we use the notation summarized in Table~\ref{table:Notation}.}

\begin{table*}
\caption{{Notation used for the secure transmission scheme over combination networks.}}
\label{table:Notation}
\begin{center}
\begin{tabular}{ |c||c| }
 \hline
 {\bf{Quantity}} & {\bf{Definition}}   \\
 \hline
\hline
{$e_i$} & {Edge from the source $S$ to the intermediate node $i \in [t]$} \\ \hline
{$X_{e_i}$} & {Symbols transmitted on edge $e_i, i \in [t]$} \\ \hline
{ $W_i$} & {Message packet for $D_i, i \in [m]$ such that $W_i := [W_i^{(1)}, W_i^{(2)}, \ldots, W_i^{(R_i)}]$} \\ \hline
{ $K_i, i \in [k]$} & {Random packet to ensure {secure} communication, with $K:= [K_1, K_2, \ldots, K_k]$} \\
 \hline
\end{tabular}
\end{center}
\vspace{-6mm}
\end{table*}

\begin{itemize}
\item {\bf{Encoding.}} { The source $S$ transmits the following symbols on its outgoing edges
	\begin{align}
\label{eq:enc}
         \begin{bmatrix}
X_{e_{1}} & X_{e_{2}} & \ldots & X_{e_{t}}  \end{bmatrix}^T
=
\begin{bmatrix} E &  V^{T} \end{bmatrix}
\begin{bmatrix} W_1 & W_2 & \ldots & W_m & K\end{bmatrix}^T,
	\end{align}
where $E$ is a matrix of $t$ rows and $\sum_{i=1}^m R_i$ columns, and $V$ is the Vandermonde matrix defined in~\eqref{eq:vandermonde}. 
Upon receiving a transmission from the source $S$, the intermediate node $i \in [t]$ simply routes this transmission on its outgoing edges, i.e., there is no coding operation at the intermediate nodes.}
	
\item {\bf{Security.}} Since any $k$ rows of $V^T$ are linearly independent, {then} any set of $k$ symbols transmitted on the first layer { -- and similarly on the second layer, since each intermediate node simply routes the received transmission on its outgoing edges --} are encoded with independent keys. {Thus, the adversary that eavesdrops $k$ edges will not be able to obtain any information about the messages, i.e.,} the scheme is secure.

\item {\bf{Decoding.}} 
{Each receiver $D_i$ will use the $R_i$ linearly independent vectors to {multiply} the vector
$         \begin{bmatrix}
X_{e_{1}} & X_{e_{2}} & \ldots & X_{e_{t}}  \end{bmatrix}^T$, and retrieve the $R_i$ private messages it requests.
Note that, because of the null spaces construction, each receiver only observes the symbols it actually has access to, and each receiver is able to cancel out the keys. }
{In Appendix~\ref{app:decoding}, we prove} that there {exists} a choice {of the matrix $E$ in~\eqref{eq:enc},} which {ensures that} all the destinations reliably decode their {intended messages.}	
\end{itemize}
\end{proof}

{
\subsection{On the Optimality of the Designed Secure Transmission Scheme}
We now conclude this section with a discussion on the rate performance achieved by the proposed secure transmission scheme. Towards this end, we start by noting that, for a combination network with parameters $(t,m, \{\mathcal{M}_i, \forall i \in [m]\} )$,} the min-cut capacity between the source $S$ and the set of destinations $D_{\mathcal{A}}:= \{D_i: i \in \mathcal{A}\}$ is given by $M_{\mathcal{A}} = \left| \bigcup\limits_{i \in \mathcal{A}} \mathcal{M}_i  \right|$. {By substituting this inside~\eqref{eq:SecureOB} in Theorem~\ref{thm:SecureOB}, we get the following outer bound,}
\begin{cor}
	\label{thm:SecureOB_combi}
	An outer bound on the secure capacity region for the multiple unicast traffic over the {combination network with parameters $(t,m, \{\mathcal{M}_i, \forall i \in [m]\} )$} is { given by} 
	\begin{align}
	R_{\mathcal{A}} \leq \left[\left| \bigcup\limits_{i \in \mathcal{A}} \mathcal{M}_i  \right| - k\right]^+,  \ \ \forall \mathcal{A} \subseteq {[m]} \enspace,	\label{eq:SecureOB_combi}	
	\end{align} 
	where $R_{\mathcal{A}} := \sum\limits_{i \in \mathcal{A}} R_i $.
\end{cor}
{We now prove that our designed secure transmission scheme is indeed optimal for the case of $m=2$ destinations. In other words, we show that the outer bound in Corollary~\ref{thm:SecureOB_combi} is achievable when $m=2$. Formally, we have}

\begin{prop}
	\label{thm:capacity}
	For the combination network {with parameters $(t,2, \{\mathcal{M}_i, \forall i \in [2]\} )$}, the secure capacity region is given by 
\begin{subequations}
\label{eq:Cap2CombNet}
	\begin{align}
	{R_1} & \leq \left[\left|\mathcal{M}_1  \right| - k\right]^+ , \\
	{ R_2} & \leq \left[\left|\mathcal{M}_2  \right| - k\right]^+ ,\\
	{ R_1 + R_2} & \leq \left[\left|\mathcal{M}_1 \cup \mathcal{M}_2  \right| - k\right]^+.
	\end{align}
\end{subequations}
\end{prop}

\begin{proof}
{Clearly, from the result in Corollary~\ref{thm:SecureOB_combi}, the rate region in~\eqref{eq:Cap2CombNet} is an outer bound on the secure capacity region. Hence, we now need to prove that the rate region in~\eqref{eq:Cap2CombNet} is also achievable. This proof can be found in Appendix~\ref{app:2CombNetAch}.}
\end{proof}

{Although we could prove the optimality of our designed secure transmission scheme only for the case of $m=2$ destinations, we performed extensive numerical evaluations that indeed suggest that the scheme {could be} optimal even for the case of larger values of $m$. In particular, in our simulations we considered up to $m=8$ destinations and, for all the considered network configurations, we verified that the rate region achieved by our designed scheme coincides with the outer bound in~\eqref{eq:SecureOB_combi}.}

%% file: two_phase_scheme_for_generalMC.tex
\label{sec:two_phase}
We now propose the design of a secure transmission scheme for networks with arbitrary topologies and arbitrary number of destinations. This scheme consists of two phases, namely the key generation phase (in which secret keys are generated between the source and the $m$ destinations) and {the} message sending phase (in which the message packets are first {encoded} using the secret keys and then transmitted to the $m$ destinations). The corresponding achievable rate region is presented in Theorem~\ref{thm:SimpleGuarantee}.
\begin{thm}
\label{thm:SimpleGuarantee}
Let $(\hat{R}_1, \hat{R}_2, \ldots, \hat{R}_m)$ be an achievable rate $m$-tuple in the absence of the eavesdropper.
Then, the rate $m$-tuple $(R_1, R_2,  \ldots, R_m)$ with
\begin{align}
\label{eq:rate2phas}
R_i = \hat{R}_i \left (1 - \frac{k}{M} \right), \forall i \in {[m]}\enspace,
\end{align}
where $M$ is the minimum min-cut between the source and any destination, is securely achievable in the presence of an {adversary who eavesdrops} any $k$ edges of her choice.
\end{thm}

\begin{proof}
Let {$M_{\{i\}}$} be the min-cut capacity between the source and the destination $D_i$ with $i \in {[m]}$.  
We define $M$ as the minimum among all these individual min-cut capacities, i.e.,  $M = \min_{i \in {[m]}} M_{\{i\}}$.
Let $(\hat{R}_1, \hat{R}_2, \ldots, \hat{R}_m) \in \mathbb{R}^m$ be the unsecure rate $m$-tuple achieved in the absence of the eavesdropper. 
We start by approximating this rate $m$-tuple with rational numbers; notice that this is always possible since the set of rationals $\mathbb{Q}$ is dense in $\mathbb{R}$.
Moreover, an information flow through the network (from the source $S$ to an artificial destination $D^\prime$ connected to all the destinations $D_i, i \in {[m]}$ -- see also {Appendix~\ref{app:single_source_unsecure}})
that achieves this rate $m$-tuple might involve fractional flows over the edges since the rate $m$-tuple may be fractional. To make the rate $m$-tuple {integer} and thereby also the flow over each edge, we multiply the capacity of each edge by a common factor $T$, {which is the least common multiple among the denominators of all the fractional flows.}
This implies that to achieve $(\hat{R}_1, \hat{R}_2, \ldots, \hat{R}_m)$, {then} $(T\hat{R}_1, T\hat{R}_2, \ldots, T\hat{R}_m)$ {has to be} achieved over $T$ instances of the network after which the flow over each edge is an integer.
In what follows, we describe our coding scheme and show that 
\begin{align}
\label{eq:RateTupleTS}
({R}_1, {R}_2, \ldots, {R}_m) \!=\!\! \left(1\!-\!\frac{k}{M} \right)\! (\hat{R}_1, \hat{R}_2, \ldots, \hat{R}_m)
\end{align}
is {securely} achievable.
This particular scheme consists of the following two phases.
\begin{itemize}
\item {\it Key generation.} This first phase -- in which secure keys are generated between the source and the destinations -- consists of $k$ subphases.
In each subphase, the source multicasts $M-k$ random packets securely to all destinations. This is possible thanks to the secure network coding result of~\cite{cai2002secure}, since the minimum min-cut capacity is $M$ and Eve has access to $k$ edges.  Thus, at the end of this phase, a total of $Tk(M-k)$ secure keys are generated, since in each phase we use the network $T$ times. 

\item {\it Message sending.} We choose $Tk$ packets out of the $Tk(M-k)$ securely shared (in the key generation phase) random packets. For each choice of $Tk$ packets, we convert the unsecure scheme achieving $(T\hat{R}_1, T\hat{R}_2, \ldots, T\hat{R}_m)$ to a secure scheme achieving the same rate $m$-tuple. Towards this end, we expand the $Tk$ shared packets into $\sum_{j=1}^m T\hat{R}_j$ packets using an MDS code matrix. With this, we have the same number of random packets as the message packets.  
We then {encode} the message packets with the random packets
and transmit them as it was done in the corresponding unsecure scheme.  
We repeat this process until we run out of the shared random packets, i.e., we repeat this process $M-k$ times by using $T$ instances of the network each time.
\end{itemize}

\noindent
{\it{Proof of security.}} 
We know that, in the absence of security considerations, a time-sharing based scheme is optimal (i.e., capacity achieving) for the multiple unicast traffic over networks with single source, i.e., network coding is not beneficial~\cite{KoetterTON2003} (see also {Appendix~\ref{app:single_source_unsecure}}).
Given that we are not using network coding operations and since each edge carries an integer information flow, then the eavesdropper will be able to wiretap at most $Tk$ different messages each {encoded} with an independent key.
Hence, the eavesdropper will not be able to obtain any information about any of the $m$ messages.

\noindent	 
{\it{Analysis of the achieved rate $m$-tuple.}}
The secure scheme described above requires a total of $M$ phases. 
In particular, in the first $k$ phases we generate the secure keys and in the remaining $M-k$ phases we securely transmit at rates of $(T\hat{R}_1, T\hat{R}_2, \ldots, T\hat{R}_m)$, over $T$ network instances.
Thus, the achieved secure message rate $(R_1, R_2, \ldots, R_m)$ is
\begin{align}
\label{eq:rateTS}
R_j  & = \frac{M-k}{M} \hat{R}_j =  \left( 1-\frac{k}{M}\right ) \hat{R}_j, \forall j \in {[m]}\enspace.
\end{align}
This concludes the proof of Theorem~\ref{thm:SimpleGuarantee}.
\end{proof}

%% file: complexityMC.tex
The capacity achieving scheme for $m=2$ destinations that we have proposed (see Section~\ref{sec:Ach2Dest}) first requires that we edge-partition the original graph $\mathcal{G}$ into three graphs (i.e., an edge in $\mathcal{G}$ appears in only one of these three graphs).
At this stage, this step requires an exhaustive search over all possible paths in the network, which requires an exponential number of operations in the number of nodes. It therefore follows that the scheme proposed in Section~\ref{sec:Ach2Dest}, even though it allows to characterize the secure capacity region, {could be} of exponential complexity.

Differently, the two-phase scheme proposed in {this section} runs in polynomial time {in the network size.} This is because all the operations that it requires (i.e., finding a $T$ such that over $T$ instances all flows are integer, multicasting the keys in the key generation phase, encrypting messages at the source {(i.e., encoding the messages with the keys)} and routing the encrypted messages) can be performed in polynomial time in the number of edges {and nodes in the network.}

The two-phase scheme described in {this section} is sub-optimal and does not achieve the outer bound in~\eqref{eq:SecureOB}. However, this scheme offers a guarantee on the {secure rate region} that can always be achieved as a function of any rate $m$-tuple that is achievable in the absence of the eavesdropper Eve (see~\eqref{eq:rate2phas} in Theorem~\ref{thm:SimpleGuarantee}). {In what follows, we seek to identify some of the reasons for which this scheme is sub-optimal.}

\begin{figure}[t]
\centering
\subfigure[]{
\includegraphics[width=0.37\textwidth]{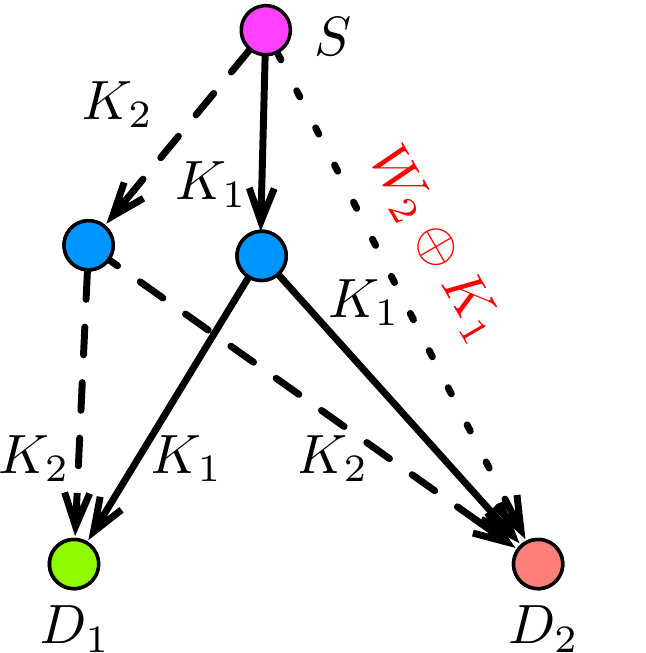}
\label{fig:cs2}
}
\hfill
\subfigure[Rate region for the network in Fig.~\ref{fig:cs2}.]{
\includegraphics[width=0.58\textwidth]{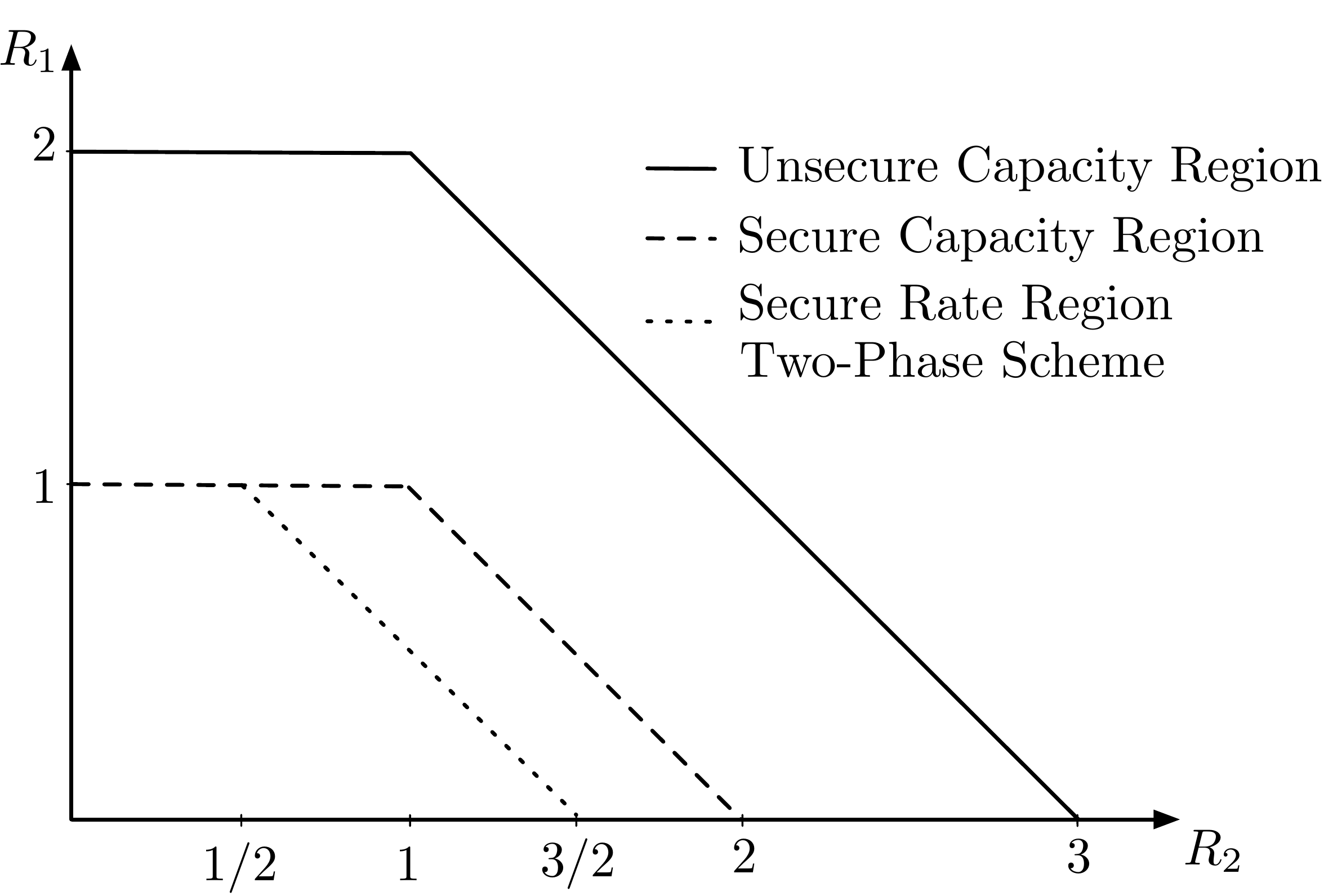}%
\label{fig:rate}
}
\caption{
Network example for which the two-phase scheme is not optimal.}
\end{figure}

One reason behind the sub-optimality is that in the key generation phase some edges in the network are not used.
Indeed, when we multicast the $M$ random packets to generate the keys (where $M$ is the minimum of the min-cut capacities and $k$ is the number of edges wiretapped by the eavesdropper) -- out of which $M-k$ linear combinations are secure keys -- it might have been possible to use the other edges (i.e., those through which the random packets do not flow) to transmit some encrypted message packets.
For instance, consider the network example in Fig.~\ref{fig:cs2}, where the eavesdropper wiretaps $k=1$ edge of her choice.
Our two-phase scheme would multicast $M = \min\limits_{i \in {[2]}} M_{\{i\}} =2$ random packets $K_1$ and $K_2$ ($K_1$ is transmitted over the solid edges and $K_2$ over the dashed edges in Fig.~\ref{fig:cs2}), out of which $M-k=1$ is securely received by $D_1$ and $D_2$. Hence, the combination $K_1 \oplus K_2$ can be used to securely transmit the message packets. However, we see that in the first phase the dotted edge (i.e., the one that connects $S$ directly to $D_2$) is not used. This brings to a reduction in the achievable rate region since this edge could have been used to securely transmit a message packet to $D_2$ by using $W_2 \oplus K_1$ as shown in Fig.~\ref{fig:cs2}.
Given this, we believe that one reason that makes the two-phase scheme suboptimal is the fact that it does not fully leverage all the network resources.
In Fig.~\ref{fig:rate}, we plotted different rate regions for the network in Fig.~\ref{fig:cs2}, which has min-cut capacities $M_{\{1\}}=2$, $M_{\{2\}}=3$ and $M_{\{1,2\}}=3$.
In particular, the region contained in the solid curve is the unsecure capacity region (given by~\eqref{eq:UnsecureOB} in Lemma~\ref{thm:UnCapacity}), the region inside the dashed curve is the secure capacity region (given by~\eqref{eq:Ach2} in Theorem~\ref{thm:SecureLB2}) and the region contained inside the dotted {curve} is the secure rate region that can be achieved by the two-phase scheme (given by~\eqref{eq:rate2phas} in Theorem~\ref{thm:SimpleGuarantee}).

\begin{figure}[t]
	\centering
	\subfigure[]{
		\includegraphics[width=0.33\textwidth]{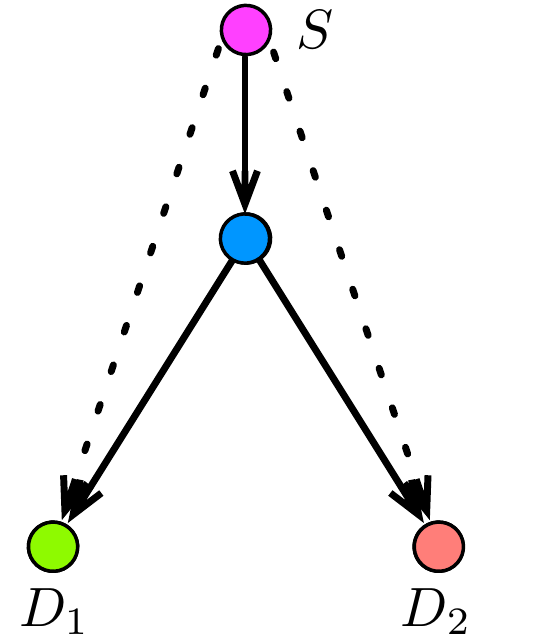}
		\label{fig:cs3}
	}
	\hfill
	\subfigure[Rate region for the network in Fig.~\ref{fig:cs3}.]{
		\includegraphics[width=0.62\textwidth]{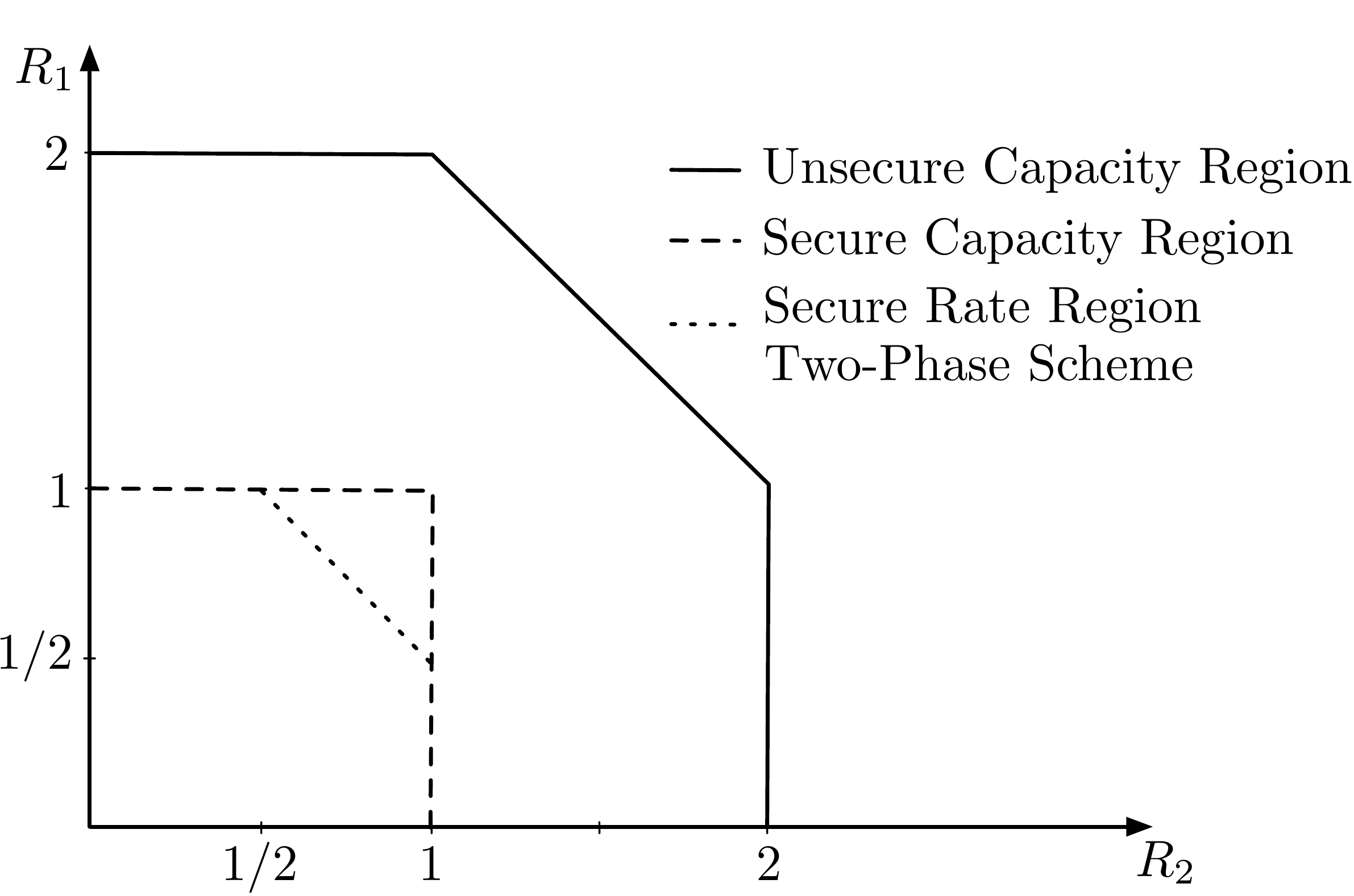}%
		\label{fig:rate3}
	}
	\caption{Network example for which the two-phase scheme is not optimal.}
	
\end{figure}

{Another} reason behind the sub-optimality is that not all {the} edges are suitable for multicasting keys. 
In particular, {tree-like structures} are more suitable to multicasting keys rather than disjoint paths to different destinations. 
To see this, we consider the network shown in Fig.~\ref{fig:cs3} {with} two destinations. 
This network has one tree structure {(represented by the solid edges in Fig.~\ref{fig:cs3})} and one set of disjoint paths to both the destinations {(represented by the two dashed edges in Fig.~\ref{fig:cs3})}. 
In { the two-phase} scheme we use both {the tree structure and the two disjoint paths} to transmit keys as well as messages, while in the optimal scheme we use the tree structure to transmit keys and the disjoint {paths} to transmit messages. In Fig.~\ref{fig:rate3}, we plotted different rate regions for the network in Fig.~\ref{fig:cs3}, which has min-cut capacities $M_{\{1\}}=2$, $M_{\{2\}}=2$ and $M_{\{1,2\}}=3$. The region contained in the solid curve is the unsecure capacity region (given by~\eqref{eq:UnsecureOB} in Lemma~\ref{thm:UnCapacity}), the region inside the dashed curve is the secure capacity region (given by~\eqref{eq:Ach2} in Theorem~\ref{thm:SecureLB2}) and the region contained inside the dotted {curve} is the secure rate region that can be achieved by the two-phase scheme (given by~\eqref{eq:rate2phas} in Theorem~\ref{thm:SimpleGuarantee}).

From Fig.~\ref{fig:rate} and Fig~\ref{fig:rate3}, we indeed observe that the rate region achieved by the two-phase scheme is contained inside the secure capacity region.
We can have a {more complete} comparison for networks with $2$ destinations. 
{For instance, consider networks for which the} min-cut capacities to both the destinations are {the} same. {Then,} depending on the number of edges {that} the adversary is eavesdropping, the {capacity} region and the region achieved using the {two-phase} scheme are shown in Fig.~\ref{fig:comp_opt_two_phase_m_2}. These {figures} are drawn {by} using Theorem~\ref{thm:SecureLB2} {(regions inside the dashed curve)} and Theorem~\ref{thm:SimpleGuarantee} {(regions inside the dotted curve)}. Unsecure capacity results {(regions inside the solid curve)} are obtained from Lemma~\ref{thm:UnCapacity}.

\begin{figure}[t]
	\centering
	\subfigure[Case 1: When $M_{\{1,2\}}^\star {\leq} k$.]{
		\includegraphics[width=0.47\textwidth]{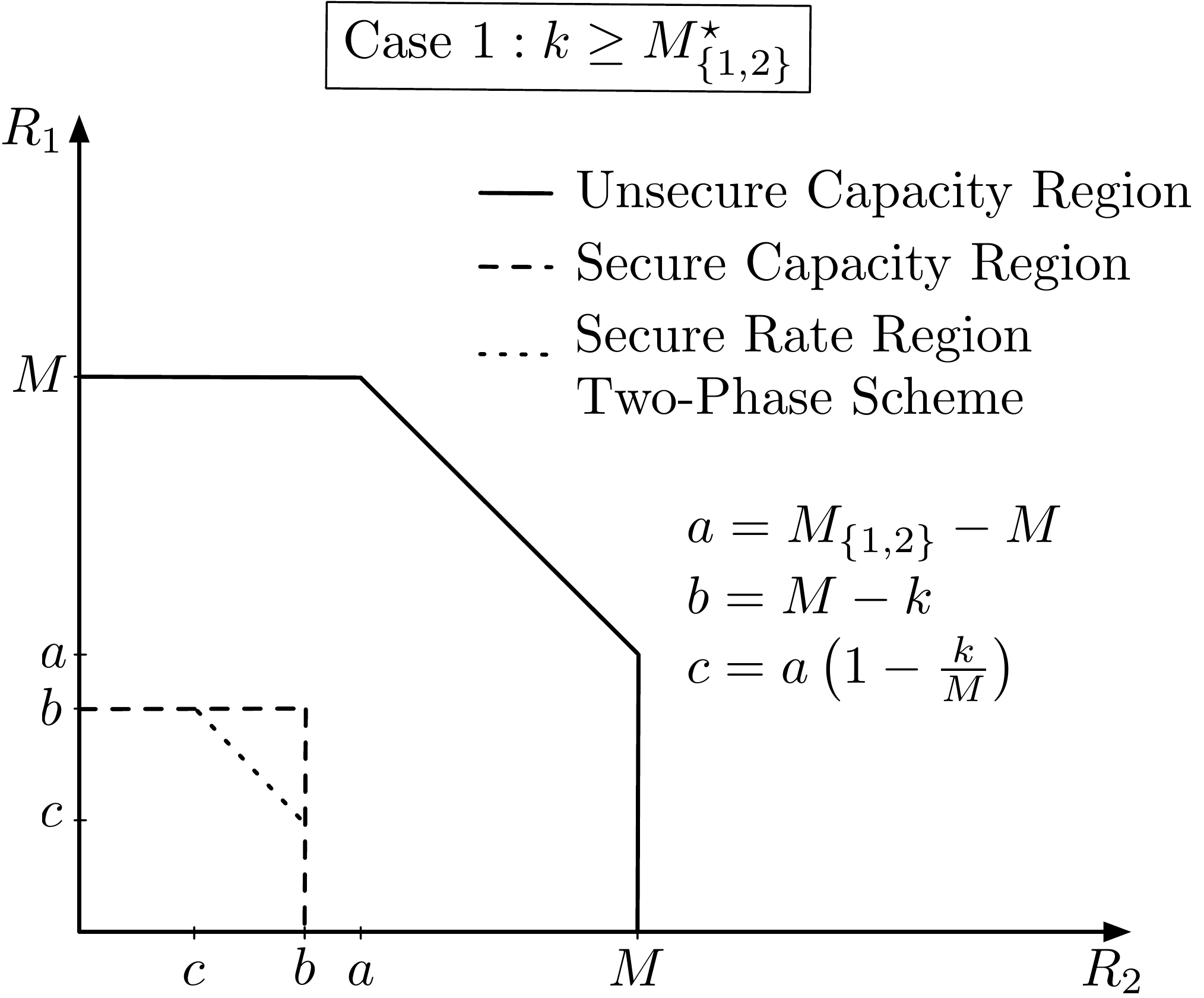}
		\label{fig:comp_opt_two_phase_m_2_case1}
	}
	\hfill
	\subfigure[Case 2: When $M_{\{1,2\}}^\star > k $.]{
		\includegraphics[width=0.47\textwidth]{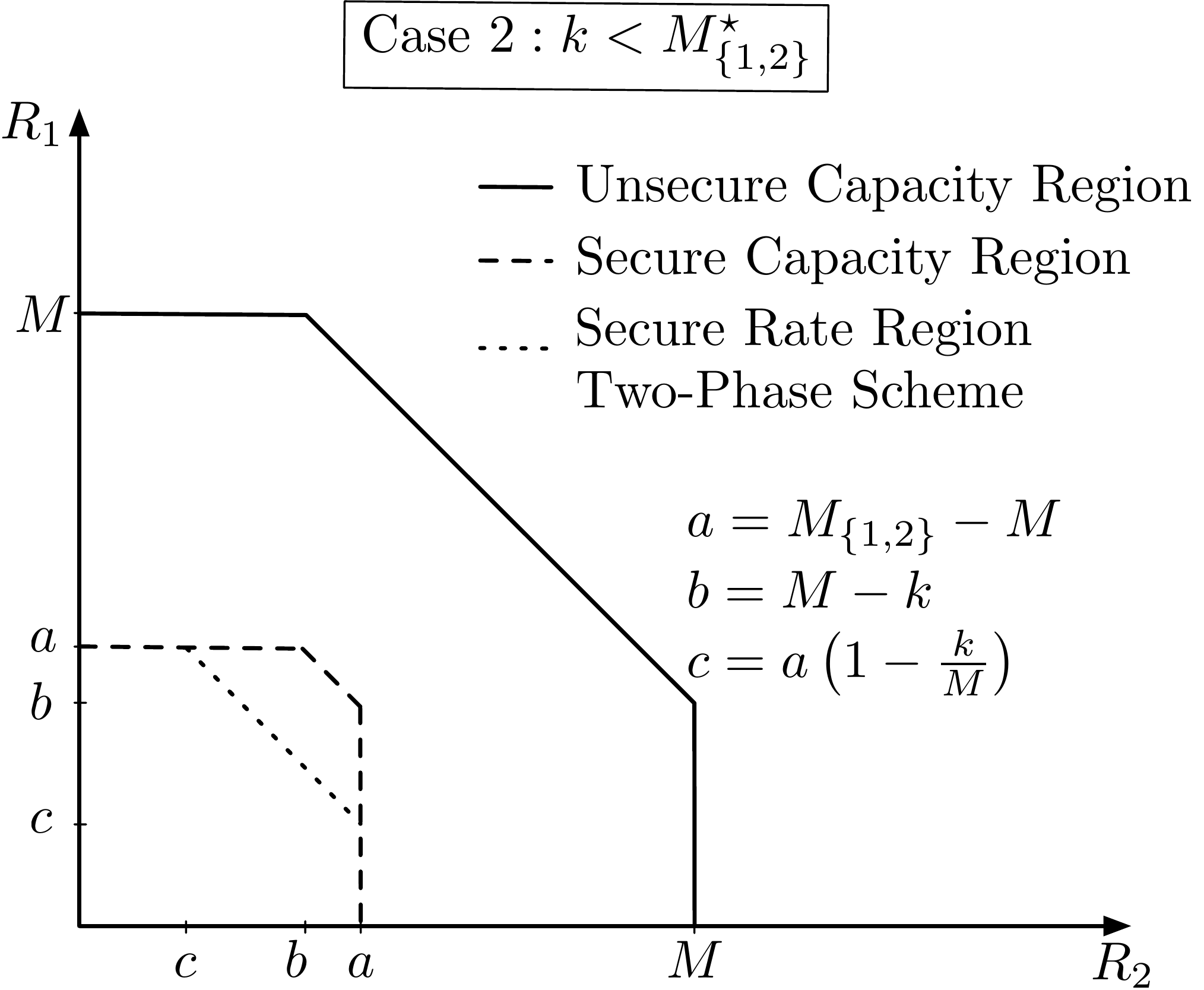}%
		\label{fig:comp_opt_two_phase_m_2_case2}
	}
	\caption{Comparison of secure capacity with the rate region achieved {by the two-phase} scheme for networks with two destinations.}
	\label{fig:comp_opt_two_phase_m_2}		
\end{figure}

%% file: single_source_unsecure_comparisionMC.tex
\label{sec:CompUnsRate}
The {unsecure capacity region} (i.e., capacity in the absence of the eavesdropper) for {a multiple unicast network with} a single source and multiple destinations described in Section~\ref{sec:setting}, is well known~\cite[Theorem 9]{KoetterTON2003} and given by the following {lemma.}
For completeness we report the proof of the following lemma in {Appendix~\ref{app:single_source_unsecure}.}

\begin{lem}
	\label{thm:UnCapacity}
	The unsecure capacity region for the multiple unicast traffic over networks with single source node and $m$ destination nodes is given by
	\begin{align}
	\label{eq:UnsecureOB}
	R_{\mathcal{A}} \leq M_{\mathcal{A}},  \ \ \forall \mathcal{A} \subseteq {[m]}\enspace,
	\end{align} 
	where $R_{\mathcal{A}} := \sum\limits_{i \in \mathcal{A}} R_i $ and $M_{\mathcal{A}}$ is the min-cut capacity between the source $S$ and the set of destinations $D_{\mathcal{A}}:= \{D_i: i \in \mathcal{A}\}$.
\end{lem}

For networks with $m=2$ destinations, we compare the secure capacity region in Theorem~\ref{thm:SecureLB2} and the unsecure capacity region in Lemma~\ref{thm:UnCapacity}.
By comparing~\eqref{eq:Ach2} with~\eqref{eq:UnsecureOB} (evaluated for the case $m=2$), we observe that in the presence of the eavesdropper we lose at most a rate $k$ in each dimension compared to the unsecure case. 
We notice that the same result holds for the case of $m=1$ destination and for the case of multicasting the same message to all destinations~\cite{cai2002secure} (i.e., we have a rate loss of $k$ with respect to the min-cut capacity $M$). However, here it is more surprising since the messages to the $m=2$ destinations (and potentially the keys) are different.

%% file: reversibilityMC.tex
 \label{sec:non_reversibility} 
In order to characterize the unsecure capacity region in~\eqref{eq:UnsecureOB}, network coding is not necessary and routing is sufficient (see also {Appendix~\ref{app:single_source_unsecure}}). Thus, from the result in~\cite{riis2007reversible}, it directly follows that the capacity result in~\eqref{eq:UnsecureOB} is reversible. 
In particular, let $\mathcal{G}$ be a network with single source and $m$ destinations with a certain capacity region (that can be computed from Lemma~\ref{thm:UnCapacity}). 
Then, the reverse graph $\mathcal{G}^\prime$ is constructed by switching the role of the source and destinations and by reversing the directions of the edges.
Thus, $\mathcal{G}^\prime$ will have $m$ sources and one single destination.
The result in~\cite{riis2007reversible} ensures that $\mathcal{G}$ and $\mathcal{G}^\prime$ will have the same capacity region, i.e., the result in Lemma~\ref{thm:UnCapacity} characterizes also the unsecure capacity region for the multiple unicast traffic over networks with $m$ sources and single destination.

We now focus on the secure case.	
In Section~\ref{sec:two_destinations}, we have characterized the secure capacity region for a multiple unicast network with single source and $m=2$ destinations.
In particular, Theorem~\ref{thm:SecureLB2} implies that the secure capacity region does not depend on the specific topology of the network and it can be fully characterized by the min-cut capacities $M_{\{1\}}, M_{\{2\}}$ and $M_{\{1,2\}}$ and by the number $k$ of edges eavesdropped by Eve.
We now show that this result is {non-reversible,} i.e., the secure capacity region of the reverse network is not the same as the one of the original network.
Moreover, we also show that the secure capacity region of networks with $2$ sources and single destination cannot anymore be characterized by only the min-cut capacities, i.e., it depends on the specific network topology.

Consider the three networks in Fig.~\ref{fig:3NetEx} and assume $k=1$, i.e., Eve wiretaps one edge of her choice. For the network in Fig.~\ref{fig:cs} we have min-cut capacities $\left (M_{\{1\}}, M_{\{2\}}, M_{\{1,2\}} \right ) = (1,2,2)$ and hence from Theorem~\ref{thm:SecureLB2} it follows that the secure capacity for this network is given by $(R_1,R_2) = (0,1)$. This point can be achieved by simply using the scheme shown in Fig.~\ref{fig:cs}, where $K$ represents the key and $W_2$ the message for $D_2$.
Now, consider the network in Fig.~\ref{fig:cd} that is obtained from Fig.~\ref{fig:cs} by switching the role of the source and destinations and by reversing the directions of the edges.
For this network, which has the same min-cut capacities as the network in Fig.~\ref{fig:cs}, 
the rate pair $(R_1,R_2) = (1,0)$ is securely achievable 
using the scheme shown in Fig.~\ref{fig:cd} where $W_1$ is the message of $S_1$ and $K_1$ and $K_2$ are the keys generated by $S_1$ and $S_2$, respectively.
The rate pair $(R_1,R_2) = (1,0)$, which is securely achieved by the network in Fig.~\ref{fig:cd}, cannot be securely achieved by the network in Fig.~\ref{fig:cs}. 
This result implies that a secure rate pair that is feasible for one network might not be feasible for the reverse network, i.e., the secure capacity regions can be different and hence cannot be derived from one another. 
The achievability of the pair $(R_1,R_2) = (1,0)$ in Fig.~\ref{fig:cd} also shows that the outer bound in~\eqref{eq:SecureOB} does not hold for networks with single destination and multiple sources, in which case it is possible to achieve rates outside this region. 

\begin{figure}
\centering
\subfigure[$(R_1,R_2)=(0,1)$ is capacity.]{
\includegraphics[width=0.343\textwidth]{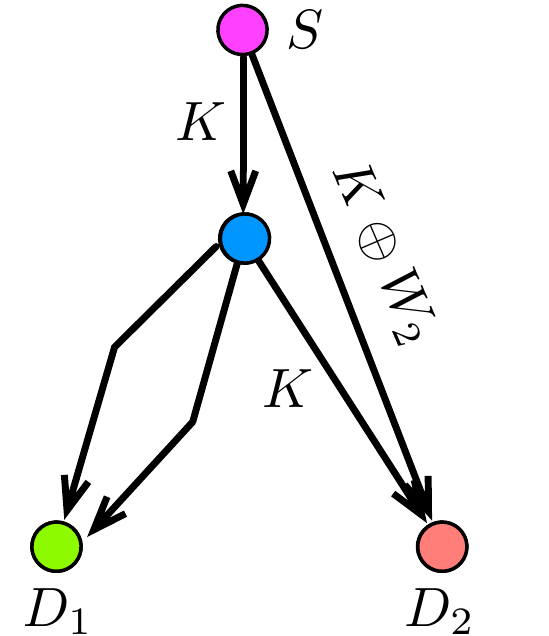}
\label{fig:cs}
}
\hfill
\subfigure[$(R_1,R_2)=(1,0)$ is achievable.]{
\includegraphics[width=0.3\textwidth]{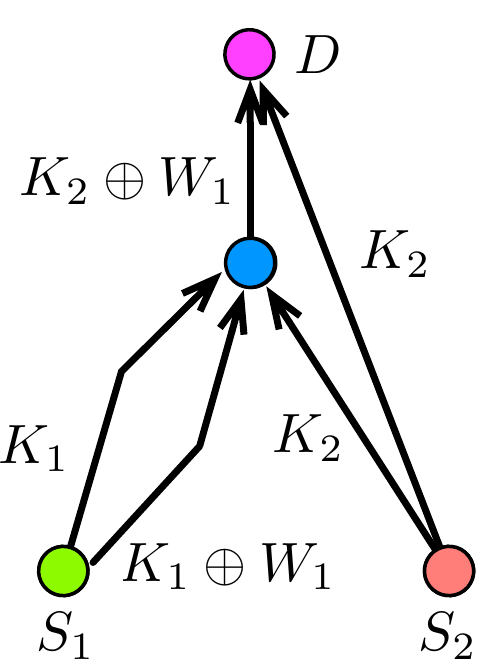}%
\label{fig:cd}
}
\hfill
\subfigure[$(\!R_1,R_2\!)\!=\!(1,0)$ is not achievable.]{
\includegraphics[width=0.28\textwidth]{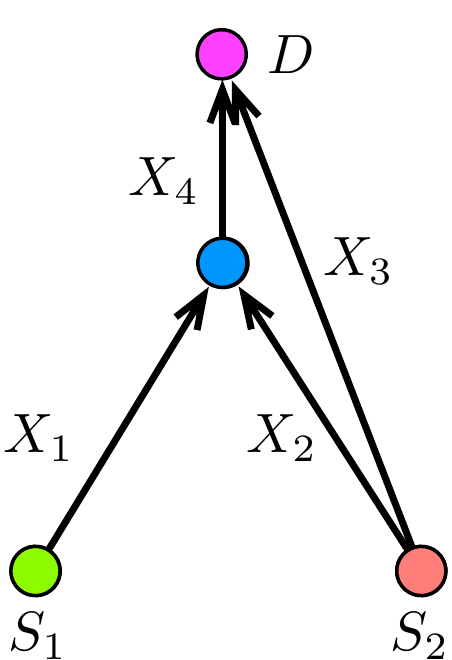}%
\label{fig:cd2}
}
\caption{Network examples {for non-reversibility.}}
\label{fig:3NetEx}
\end{figure}

Consider now the network in Fig.~\ref{fig:cd2}. This network has the same min-cut capacities {as the network in Fig.~\ref{fig:cd}, i.e., $\left (M_{\{1\}}, M_{\{2\}}, M_{\{1,2\}} \right )$ = $(1,2,2)$.} We now show that the rate pair $(R_1,R_2) = (1,0)$, which can be securely achieved in the network in Fig.~\ref{fig:cd}, cannot be securely achieved in the network in Fig.~\ref{fig:cd2}. Let $X_i, i \in [1:4],$ be the transmitted symbols as shown in Fig.~\ref{fig:cd2}. With this, we have 
\begin{align*}
R_1  = H(W_1) 
 & \stackrel{{\rm{(a)}}}{=} H(W_1) - H(W_1|X_3,X_4)  \stackrel{{\rm{(b)}}}{\leq} H(W_1) - H(W_1 | X_1, X_2, X_3) \\ 
& = I(W_1; X_1, X_2, X_3)  = I(W_1; X_1) + I(W_1; X_2, X_3 | X_1) \\
& \stackrel{{\rm{(c)}}}{=} I(W_1; X_2, X_3 | X_1) = H(X_2, X_3 | X_1) - H(X_2, X_3 | W_1, X_1)  \\
& \stackrel{{\rm{(d)}}}{=} H(X_2, X_3) - H(X_2, X_3)   = 0\enspace,
\end{align*}
where: 
(i)  the equality in $\rm{(a)}$ follows because of the decodability constraint;
(ii) the inequality in $\rm{(b)}$ follows because of the `conditioning reduces the entropy' principle and since $X_4$ is a deterministic function of $(X_1, X_2)$;
(iii) the equality in $\rm{(c)}$ follows because of the perfect secrecy requirement;
(iv) finally, the equality in $\rm{(d)}$ follows since $(X_2, X_3)$ is independent of $(W_1, X_1)$.
This result shows that the rate pair $(R_1,R_2) = (1,0)$ is not securely achievable in the network in Fig.~\ref{fig:cd2}.
This implies that, for a network with single destination and multiple sources, we cannot characterize the secure capacity region based only on the min-cut capacities $\left (M_{\{1\}}, M_{\{2\}}, M_{\{1,2\}} \right )$, i.e., the result would depend on the specific network topology.

%% file: other_instancesMC.tex
\label{sec:OtherInst}

In this paper, we {have} focused on noiseless networks {with unit edge capacities having a single source and multiple destinations.} 
{We now consider other instances of multiple unicast traffic, and provide secure capacity results for some specific configurations. The main goal of this analysis is to highlight the critical role of coding across different unicast sessions in order to ensure a secure communication, even for scenarios where it is not required in the absence of an adversary.}

\begin{figure}[t]
	\centering
	\includegraphics[width=0.8\columnwidth]{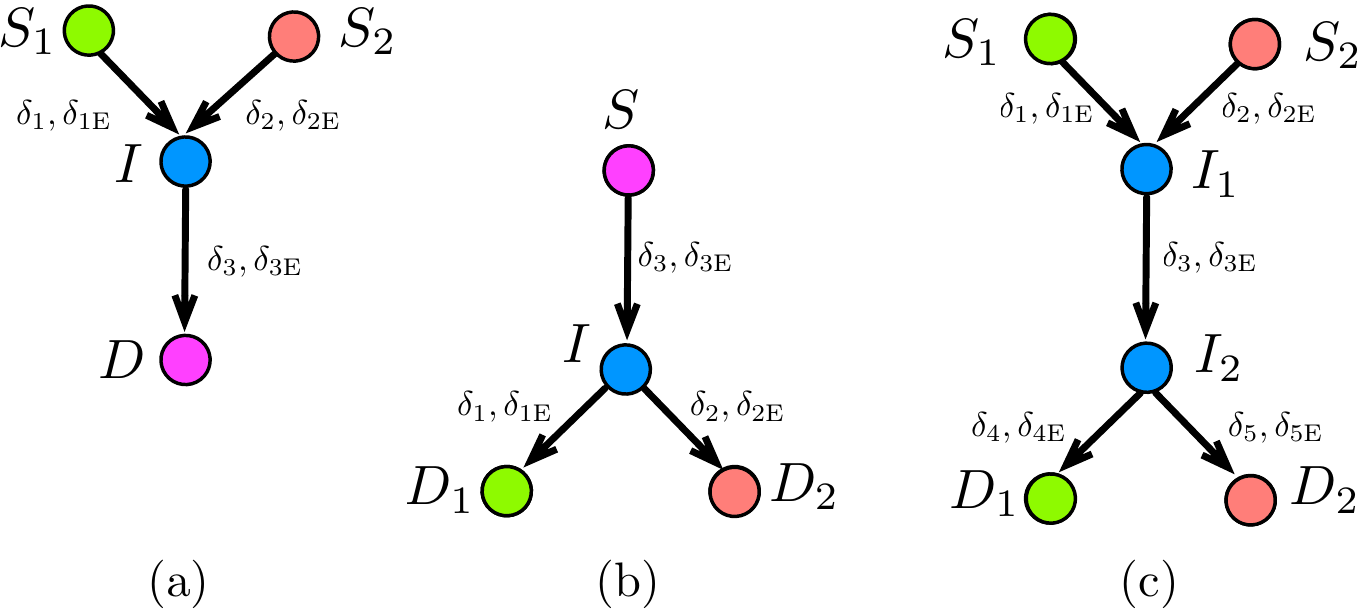}
	\caption{(a) The Y-network, (b) The RY-network and (c) The X-network.}
	\label{fig:networks}
	\vspace{-5mm}
\end{figure}

\subsubsection{Erasure Networks}
{In~\cite{agarwal2016secure},} we considered multiple unicast traffic over {the three} erasure networks {shown in Fig.~\ref{fig:networks}, referred to as} the Y-network, the Reverse Y (RY)-network and the X-network. 
In the Y-network, two sources wish to communicate two independent messages to a common destination. In the RY-network, one source aims to communicate two independent messages to two different {destinations.} Finally, in the X-network two sources seek to communicate two independent messages to two different {destinations.}
In all three cases, only {the} source(s) can generate randomness; while in Fig.~\ref{fig:networks}(a) and Fig.~\ref{fig:networks}(c) sources can generate randomness at an infinite rate, in Fig.~\ref{fig:networks}(b) the source can generate randomness only at a finite rate $D_0$. 
These assumption are motivated by the fact that one can construct the X-network by {combining} the Y-Network and the RY-network. 
{Each} edge $e$ on these three networks models an erasure channel where the legitimate receiver has an erasure probability of $\delta_e$ and the adversary has an erasure probability of $\delta_{e\text{E}}$. Public feedback, which in~\cite{Maurer} was shown to increase the {secure} capacity, is used, i.e., each of the legitimate nodes involved in the communication sends an acknowledgment after each transmission; this is received by {all nodes in the network as well as by the eavesdropper (who can wiretap any {$k=1$} channel of the network). 
{In~\cite{agarwal2016secure}, we derived the secure capacity region for the three networks in Fig.~\ref{fig:networks}, as the solution of some feasibility programs that, for completeness we report in Appendix~\ref{app:OtherNetInstEras}. In particular, our capacity-achieving secure transmission schemes consist of two phases.}
In the first phase, a link by link key is shared, and in the second phase encrypted message packets {(i.e., encoded with the keys that were generated in the first phase)} are transmitted. The key sharing mechanism involves a mix of communicating random symbols using an MDS code and an {Automated Repeat ReQuest (ARQ)} based scheme. 
{We start by noting that, in order to characterize the capacity region of the three networks in Fig.~\ref{fig:networks} in the absence of the adversary, coding is not needed and a simple time-sharing approach among the two unicast sessions is capacity-achieving. However, under security considerations, coding becomes of fundamental importance.
With the primal goal to show the benefits of coding across the two unicast sessions,} we here compare the {secure} capacity performance of our schemes {(see Propositions~\ref{thm:capY}-\ref{thm:capX} in Appendix~\ref{app:OtherNetInstEras})} with respect to two naive strategies, {i.e.,} the {\it path sharing} and the {\it link sharing}.
In the {\it path sharing} the whole communication resources, at each time instant, are used only by {one session;} 
for example, for the X-network
{we have a time-sharing between {${S}_1$-${I}_1$-${I}_2$-${D}_1$} and {${S}_2$-${I}_1$-${I}_2$-${D}_2$.}
Differently, in the {\it link sharing} strategy only the shared communication link is time-shared among the two unicast sessions; for example, in the X-network only the {${I}_1$-${I}_2$} link is time-shared.
For both these strategies we do not allow the source node that does not participate to act as a source of randomness, e.g., for the X-network the random packets sent by {${S}_1$} cannot be used to {encode} the message packets of {${S}_2$.}
Fig.~\ref{fig:NumEv} shows the performance (in terms of secure capacity region) of these two time-sharing strategies and of our schemes {(see Propositions~\ref{thm:capY}-\ref{thm:capX} in Appendix~\ref{app:OtherNetInstEras}).}
From Fig.~\ref{fig:NumEv},  we observe that our {schemes (solid line)} achieve higher rates compared to the two time-sharing strategies. 
{In general, these gains follow since:}
(i) in the Y-network {${S}_1$} and {${S}_2$} transmit random packets to {$I$} and {these can be mixed} to create a key on the shared link; 
(ii) in the RY-network the same set of  random packets can be used to generate keys for both the {$I$-$D_1$ and $I$-$D_2$} links. 
These factors, {which involve coding operations across the two sessions,} decrease the number of random packets required to be sent from the source(s) and implies that more message packets can be carried.
Finally, (iii) in the X-network we have the benefits of both the Y- and RY-network.

\begin{figure}[t]
	\centering
	\subfigure[Y-network: 
	$\left( \delta_1,\delta_{1\text{E}}\right)\!=\!\left(0.2,0.05 \right)$, 
	$\left( \delta_2,\delta_{2\text{E}}\right)\!=\!\left(0.3,0.05 \right)$,
	$\left( \delta_3,\delta_{3\text{E}}\right)\!=\!\left(0.25,0.05 \right)$.]{
		\includegraphics[width=0.47\linewidth]{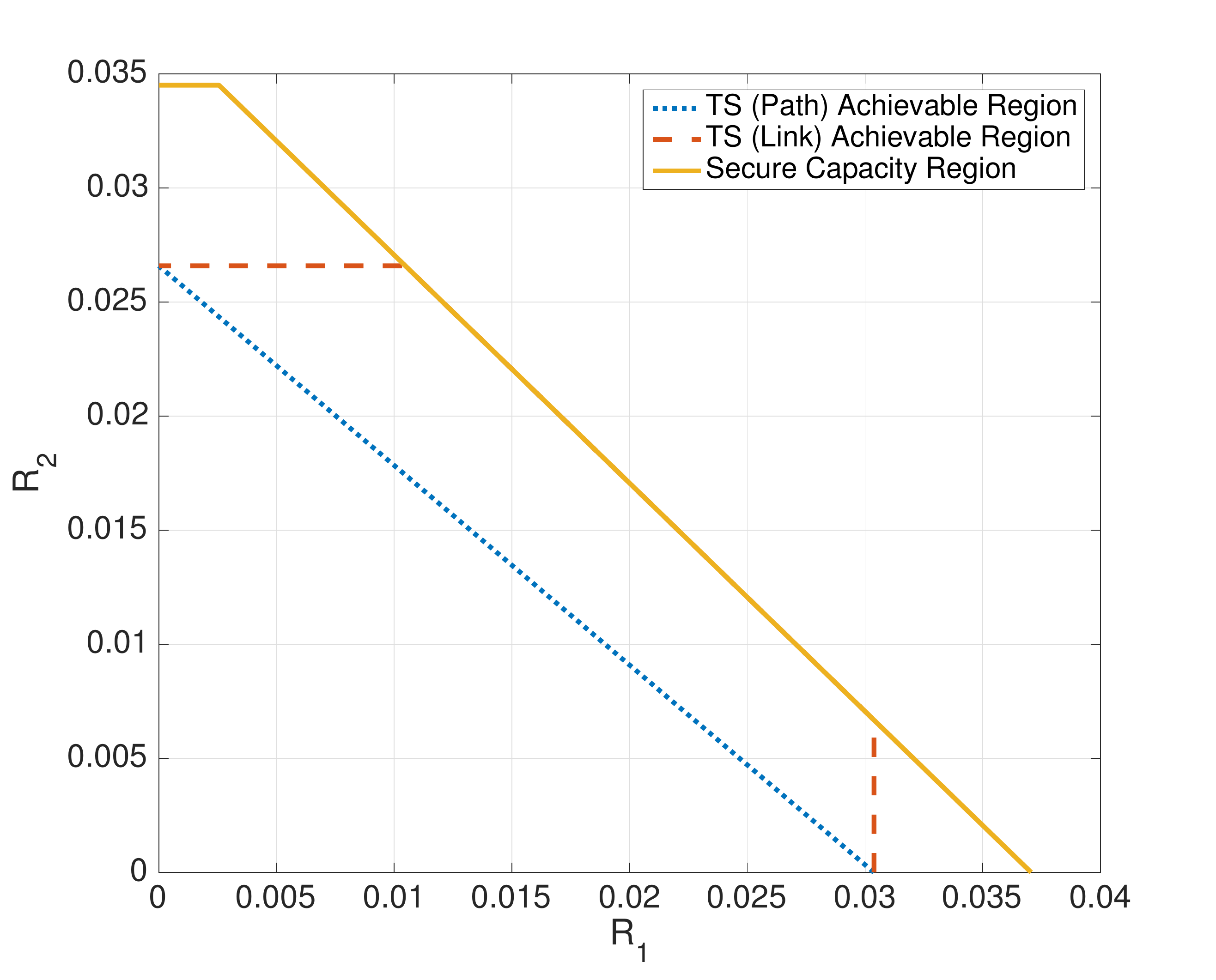}
		\label{fig:YnSim}
	}
	\hfill
	\subfigure[RY-network: 
	$\left( \delta_1,\delta_{1\text{E}}\right)\!=\!\left(0.1,0.1 \right)$, 
	$\left( \delta_2,\delta_{2\text{E}}\right)\!=\!\left(0.2,0.05 \right)$,
	$\left( \delta_3,\delta_{3\text{E}}\right)\!=\!\left(0.3,0.15 \right)$, $D_0\!=\!0.4$.]{
		\includegraphics[width=0.47\linewidth]{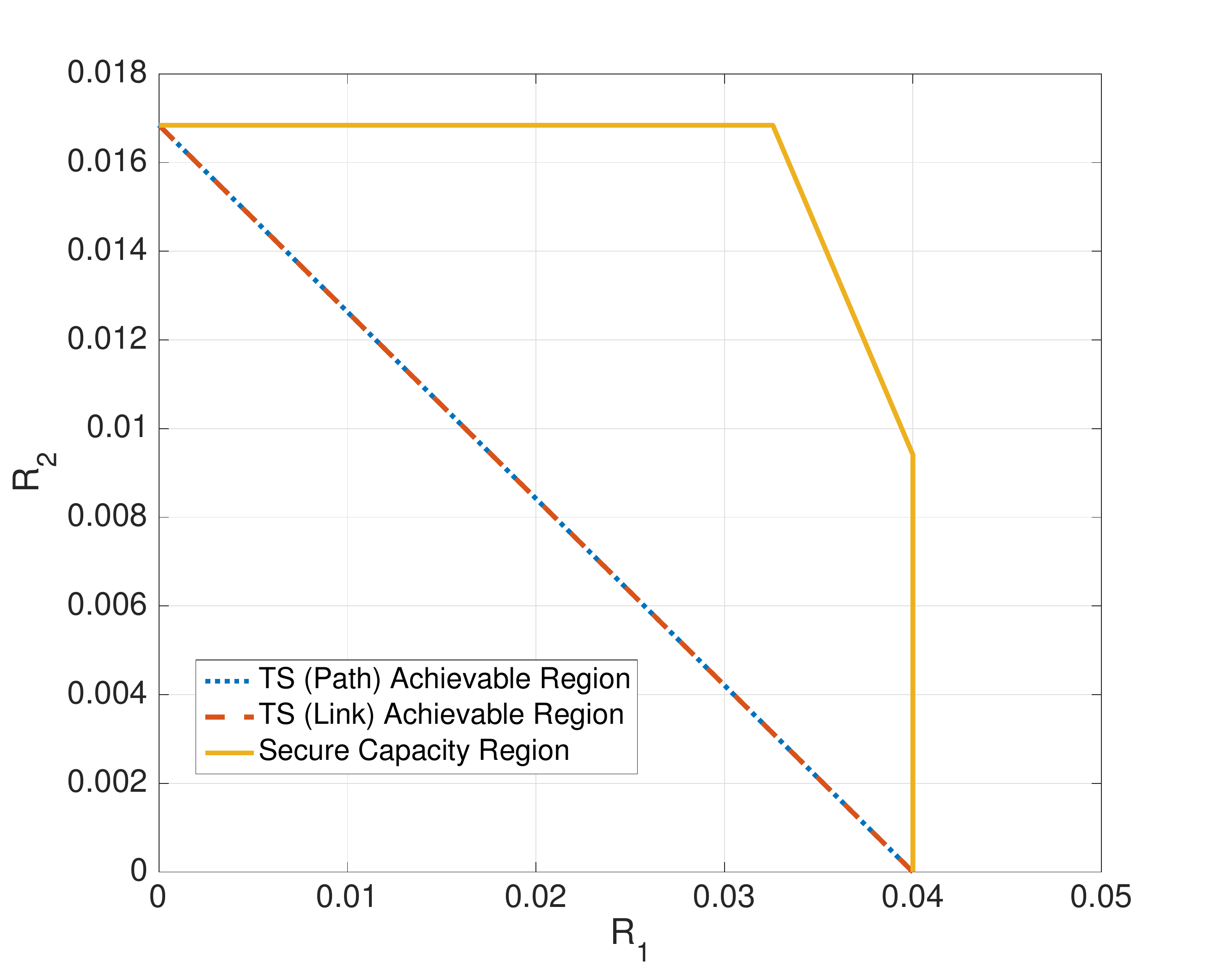}%
		\label{fig:RYnSim}
	}
	\hfill
	\subfigure[X-network:
	$\left( \delta_1,\delta_{1\text{E}}\right)\!=\!\left(0.1,0.1 \right)$, 
	$\left( \delta_2,\delta_{2\text{E}}\right)\!=\!\left(0.2,0.05 \right)$,
	$\left( \delta_3,\delta_{3\text{E}}\right)\!=\!\left(0.3,0.15 \right)$, 
	$\left( \delta_4,\delta_{4\text{E}}\right)\!=\!\left(0.4,0.25 \right)$,
	$\left( \delta_5,\delta_{5\text{E}}\right)\!=\!\left(0.5,0.2 \right)$.
	]{
		\includegraphics[width=0.47\linewidth]{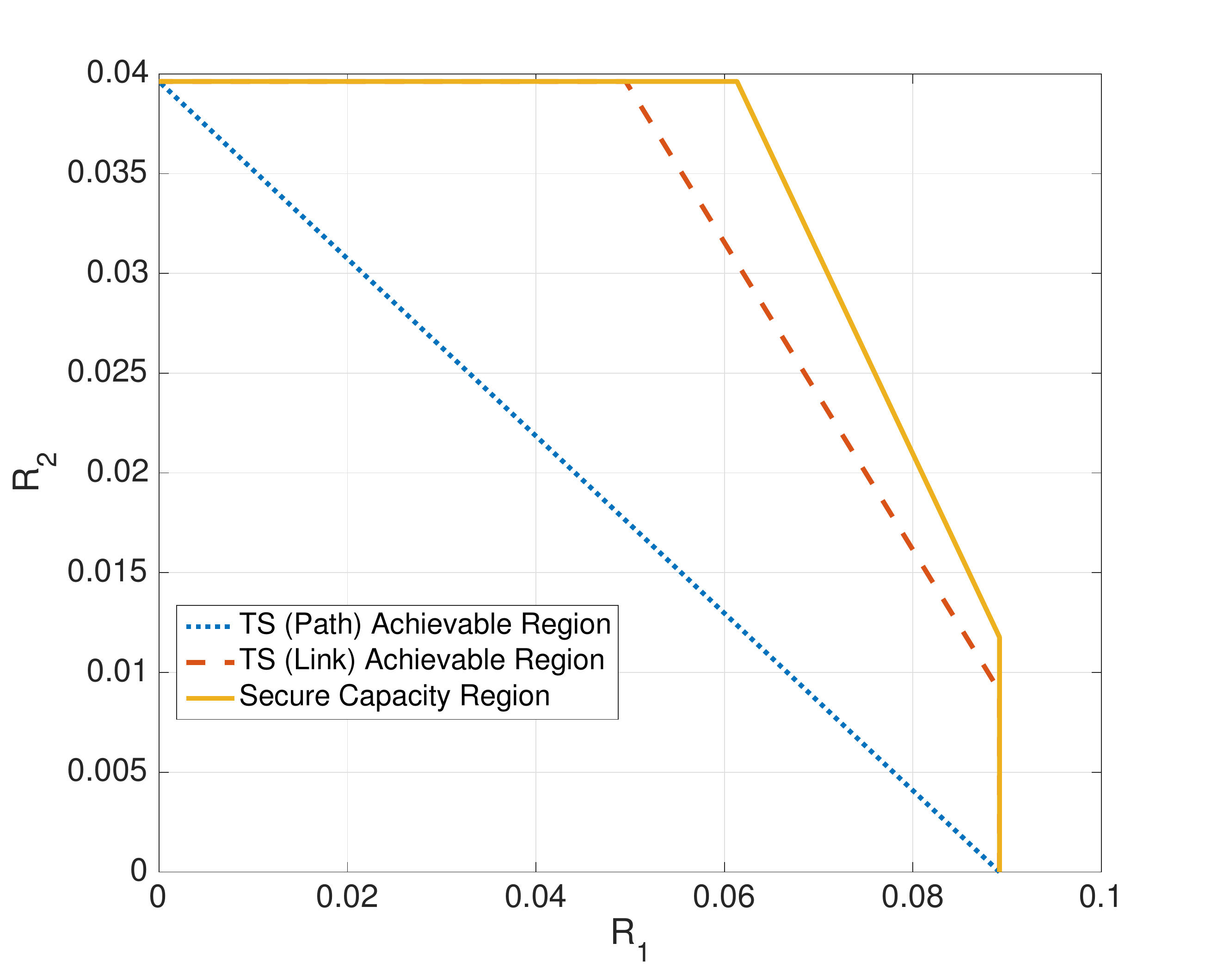}
		\label{fig:XnSim}
	}
	\vspace{-3mm}
	\caption{Numerical evaluations for the three networks in Fig.~\ref{fig:networks}.}
	\label{fig:NumEv}
	\vspace{-4mm}
\end{figure}

\begin{figure}[ht]
	\centering
	\subfigure[Butterfly network~1.]{
		\includegraphics[width=0.33\linewidth]{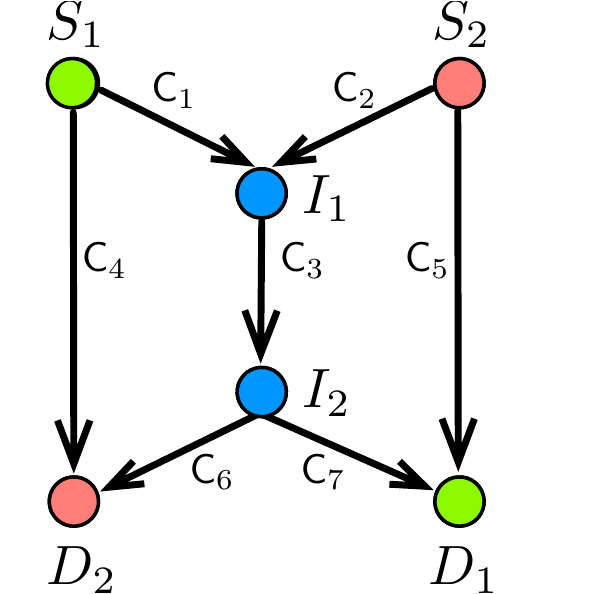}
		\label{fig:butt1}
	}
	\hspace{2cm}
	\subfigure[Butterfly network with {single source.}]{
		\includegraphics[width=0.33\linewidth]{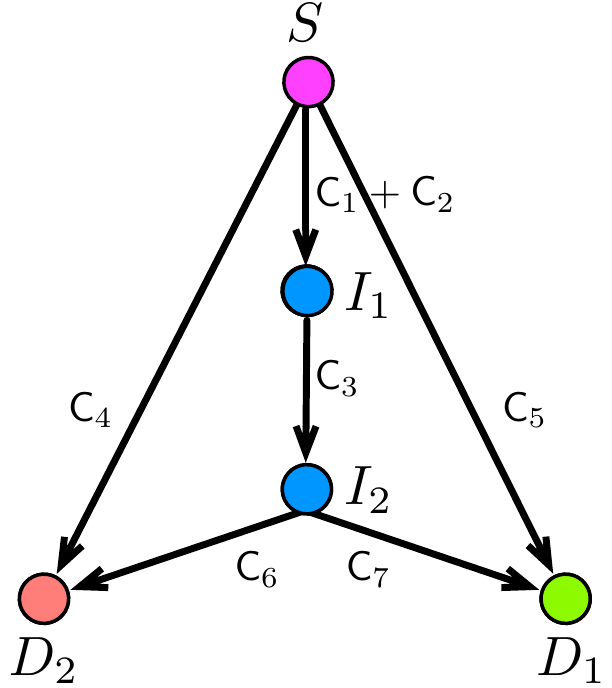}%
		\label{fig:CS}
	}
\\
	\subfigure[Butterfly network with {single destination.}]{
		\includegraphics[width=0.25\linewidth]{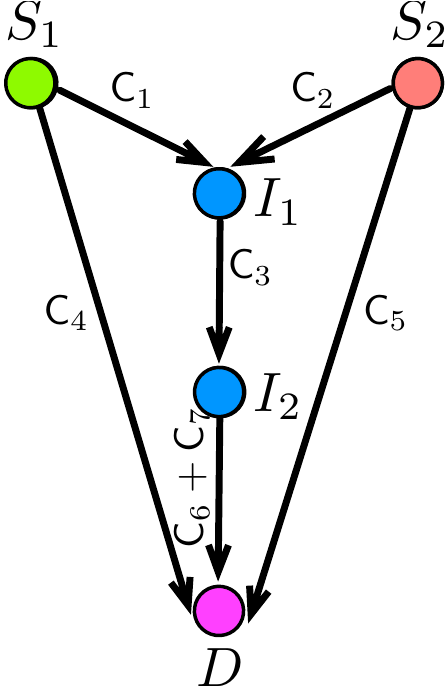}
		\label{fig:CD}
	}
	\hspace{3cm}
	\subfigure[Butterfly network~2.]{
		\includegraphics[width=0.35\linewidth]{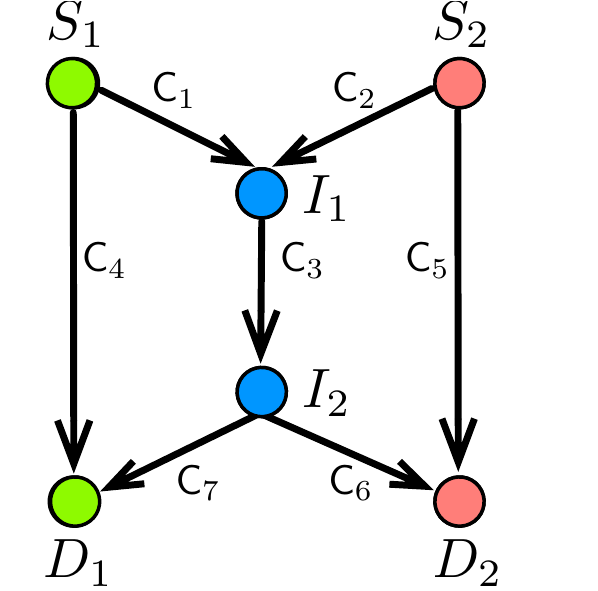}
		\label{fig:butt2}
	}
	\caption{Networks derived from the butterfly network.}
	\label{fig:butterflies}
\end{figure}

\subsubsection{Arbitrary Edge Capacities}
In~\cite{agarwal2016netcod}, we {considered the} four {noiseless} networks shown in Fig.~\ref{fig:butterflies}, {which} are derived from the celebrated butterfly network. 
{In particular, the edges can have arbitrary capacity, which represents a main difference with respect to the model analyzed in the previous sections. 
Over these networks, we assume that the passive eavesdropper can wiretap any $k=1$ edge of her choice, and that only the source(s) can generate randomness at infinite rate.
In~\cite{agarwal2016netcod}, we derived both the secure and the unsecure capacity regions for these four networks, which for completeness we report in Appendix~\ref{app:OtherNetArbrCap}. In particular, our capacity-achieving schemes show the critical importance of coding across the two unicast sessions, as opposite to a {naive time-sharing} approach.}
{For instance, consider the butterfly network with single source in Fig.~\ref{fig:CS}, with all edge capacities equal to $\mathsf{C}$ and $R_1=R_2=R$.}
{Over this network,} if we simply timeshare across the two sessions, then we get $R \leq \frac{\mathsf{C}}{2}$, i.e., {the edge between $I_1$ and $I_2$} is the bottleneck. 
However, by using the coding operation (i.e., the same secret key can be used by the two sessions), we obtain $R \leq \mathsf{C}$, i.e., we can transmit at a double rate {(see~\eqref{eq:CScapSec} in Appendix~\ref{app:OtherNetArbrCap}).}
{When} all edge capacities {are} equal to $\mathsf{C}$ and $R_1=R_2=R$, we can also draw the following additional conclusions.
\begin{itemize}
	
	\item 
	Secure communication can incur significant rate losses with respect to the unsecure {case.} 
	{The rate losses can be quantified as:}
	(i) $100\%$ for butterfly network 1 (secure communication is not possible);
	(ii) more than $33\%$ for the case of {single source} ($R \leq \frac{3}{2} \mathsf{C}$ without security and $R \leq \mathsf{C}$ with security);
	(iii) more than $66\%$ for the case of {single destination} and butterfly network 2 ($R \leq \frac{3}{2} \mathsf{C}$ without security and $R \leq \frac{\mathsf{C}}{2}$ with security).
	\item For unsecure communication, {the butterfly networks with single source and single destination achieve} a rate gain of $50\%$ over butterfly network~1. 
	This gain is due to an increase in the min-cut values, which are tight and evaluate to $R \leq \mathsf{C}$ in the butterfly network 1 and to $R \leq \frac{3}{2} \mathsf{C}$ in the cases of {single source and single destination.}
	
	\item 
	Under security considerations, the case of {single source} (i.e., $R \leq \mathsf{C}$) brings higher throughput gains than the {single destination case} ($R \leq \frac{\mathsf{C}}{2}$).
	{This is because in the former case, coding opportunities arise, i.e., the same key can be used by the two sessions.}
	Moreover, {thanks to the multipath diversity,} {these two} cases enable secure communication, which was not possible over the butterfly network 1. 
	{Regarding the butterfly network 2 the case of {single source} brings secure rate advantages.}
	Actually, for both the butterfly network 2 and the case of {single source,} the min-cut values evaluate to $R \leq \frac{3}{2} \mathsf{C}$, but the secure rate achieved in the former case, i.e., $R \leq \frac{\mathsf{C}}{2}$, is half the one achieved in the latter case, i.e., $R \leq \mathsf{C}$.
\end{itemize}

%% file: seperable_graph_m_2MC.tex
\section{Proof of Lemma~\ref{lemma:sep}}
\label{app:seperable_graph_m_2}
For completeness, we here report
the proof of the result in Lemma~\ref{lemma:sep}, which is a direct consequence of~\cite[Theorem 1]{ramamoorthy2009single}. In particular, this result shows that any graph $\mathcal{G}$ with single source and $m=2$ destinations is separable.
The graph $\mathcal{G}$ has min-cut capacity $M_{\{i\}}, i \in {[2],}$ towards destination $D_i$ and min-cut capacity $M_{\{1,2\}}$ towards $\{D_1,D_2\}$, from which $M_{\{i\}}^\star, i \in {[2],}$ and $M_{\{1,2\}}^\star$ can be computed by using the expressions in~\eqref{eq:Mstar}.
We represent these min-cut capacities by the triple 
\begin{align*}
\left(M_{\{1\}},M_{\{2\}},M_{\{1,2\}} \right) \!=\! \left(M_{\{1\}}^\star \!+\! M_{\{1,2\}}^\star,M_{\{2\}}^\star \!+\! M_{\{1,2\}}^\star,M_{\{1\}}^\star \!+\!  M_{\{2\}}^\star + M_{\{1,2\}}^\star \right)\enspace,
\end{align*}
where the equality follows by using~\eqref{eq:Mstar}.
We now prove Lemma~\ref{lemma:sep} in two steps.
We first show that the graph $\mathcal{G}$ can be separated into two graphs: $\mathcal{G}_a$ with min-cut capacities $\left(M_{\{1\}}^\star,0,M_{\{1\}}^\star \right)$ and $\mathcal{G}_b$ with min-cut capacities 
\begin{align*}
\left(M_{\{1,2\}}^\star,M_{\{2\}}^\star + M_{\{1,2\}}^\star,M_{\{2\}}^\star + M_{\{1,2\}}^\star\right)\enspace.
\end{align*}
Then, by applying the same principle we further separate the graph $\mathcal{G}_b$ into two graphs: $\mathcal{G}_c$ with min-cut capacities $\left(0,M_{\{2\}}^\star,M_{\{2\}}^\star\right)$ and $\mathcal{G}_d$ with min-cut capacities $\left(M_{\{1,2\}}^\star,M_{\{1,2\}}^\star,M_{\{1,2\}}^\star\right)$.
This would complete the proof of Lemma~\ref{lemma:sep}.

We now prove that we can separate the graph $\mathcal{G}$ into the two graphs $\mathcal{G}_a$ and $\mathcal{G}_b$.
Towards this end, from the original graph $\mathcal{G}$, we create a new directed acyclic graph $\mathcal{G}^\prime$ where a new node $D^\prime$ is connected to $D_1$ through an edge of capacity $M_{\{1\}}^\star + M_{\{1,2\}}^\star$ and to $D_2$ through an edge of capacity $M_{\{2\}}^\star$.
By following similar steps as in the proof of the direct part (achievabiliy) of Lemma~\ref{thm:UnCapacity} (see {Appendix~\ref{app:single_source_unsecure}}),
it is not difficult to see that in $\mathcal{G}^\prime$ the min-cut capacity between $S$ and $D^\prime$ is $M_{\{1\}}^\star + M_{\{1,2\}}^\star + M_{\{2\}}^\star = M_{\{1,2\}}$, where the equality follows from~\eqref{eq:M12star}.
From the max-flow min-cut theorem, we can find $M_{\{1,2\}}$ edge-disjoint paths from $S$ to $D^\prime$; we color the edges in these paths {\em{green}}. 
We can also find $M_{\{2\}}$ edge-disjoint paths from $S$ to $D_2$;  we color the edges in these paths {\em{red}}.
Notice that, at the end of this process, some of the edges can have both {\em{green}} and {\em{red}} colors. 
We also highlight that:
\begin{itemize}
\item Out of the $M_{\{1,2\}}$ {\em green} paths from $S$ to $D^\prime$, $M_{\{1\}}^\star + M_{\{1,2\}}^\star $ paths flow through $D_1$ and $M_{\{2\}}^\star$ flow through $D_2$.
\item If a path is exclusively {\em{green}}, it flows through $D_1$ since otherwise, in addition to the $M_{\{2\}}$ {\em{red}} edge-disjoint paths from $S$ to $D_2$, we would have also this path and thereby violate the min-cut capacity constraint to $D_2$.
\end{itemize}

The second observation above implies that, if there are $M_{\{1\}}^\star$ exclusively {\em{green}} paths, then we can separate the graph $\mathcal{G}^\prime$ into two graphs: $\mathcal{G}^\prime_a$ that contains all these $M_{\{1\}}^\star$ exclusively {\em{green}} paths and $\mathcal{G}^\prime_b$ that contains all the edges of $\mathcal{G}^\prime$ that are not in $\mathcal{G}^\prime_a$.
Given this, by simply removing the node $D^\prime$ and its incoming edges, we get $\mathcal{G}_a$ and $\mathcal{G}_b$.
We now show how we can obtain these $M_{\{1\}}^\star$ exclusively {\em{green}} paths.
Towards this end, we denote with $\mathcal{P}$ the set of all {\em green} paths from $S$ to $D^\prime$ (notice that these paths might have also some {\em red} edges).
Then, until there exists a path $p \in \mathcal{P}$ such that either it is not exclusively {\em{green}} or it does not start with an edge that is both {\em{red}} and {\em{green}}, we apply the two following steps:
\begin{enumerate}
\item Let $e$ be the first edge in $p$, which is both {\em green} and {\em red} and denote with $g$ the {\em red} path from $S$ to $D_2$ that contains the edge $e$. 
Recall that, since the $M_{\{2\}}$ {\em red} paths are edge-disjoint, there is only one {\em red} path $g$ passing through $e$.
We split the path $p$ into two parts as $p_1 - e - p_2$ and similarly we split the path $g$ into $g_1 - e - g_2$.
\item We add the {\em red} color to $p_1$ (that before was all {\em green}) and we remove the {\em red} color from $g_1$, i.e., now each edge in $g_1$ is either {\em green} or it does not have any color. 
Note that in this way we replace the {\em red} path $g_1 - e - g_2$ with $p_1 - e -g_2$ from source $S$ to $D_2$, which is also disjoint from the rest of $M_{\{2\}} -1 $ {\em red} paths.
\end{enumerate} 
	
We note that this process will stop only when all the $M_{\{1,2\}}$ paths from $S$ to $D^\prime$ are either exclusively green or start with an edge that is both {\em red} and {\em green}. 
We also note that, since we did not remove any edge, clearly we also did not change any min-cut capacity during this process. 
Since initially there were $M_{\{2\}}$ {\em red} edges coming out of $S$ and, in the process of the algorithm, we replaced one {\em red} by another {\em red}, then the number of {\em red} edges outgoing from $S$ still remains the same. 
Thus, among the $M_{\{1,2\}}$ paths from $S$ to $D^\prime$, only at most $M_{\{2\}}$ paths start with an edge that is both {\em green} and {\em red} and therefore, by using~\eqref{eq:Mstar}, at least $M_{\{1\}}^\star$ are exclusively {\em green} paths.
This proves that the original graph $\mathcal{G}$ can be separated into the two graphs $\mathcal{G}_a$ and $\mathcal{G}_b$. By using similar arguments, one can then show that the graph $\mathcal{G}_b$ can be separated into the two graphs $\mathcal{G}_c$ and $\mathcal{G}_d$. This concludes the proof of Lemma~\ref{lemma:sep}.

%% file: decodingMC.tex
\section{Proof of existence of $E$ in~\eqref{eq:enc} for a reliable decoding}
\label{app:decoding}

The destination $D_i$ will {receive symbols} $\{X_{e_j}, \ j \in \mathcal{M}_i \}$. Further, {since $N_i$ is the right null space of $V_i$ in~\eqref{eq:vandermonde_i}, then} any vector {that belongs to} $N_i$ will have {non-zero components only in the positions indexed by $\{i \in \mathcal{M}_i\}$.}
{It therefore follows that} an inner product between a vector in $N_i$ and ${[X_{e_1}, \ldots, X_{e_m}]}$ is a valid decoding scheme. 
	
	Let $d_i^{(1)}, d_i^{(2)}, \ldots, d_i^{(R_i)}$ be the $R_i$ {column} vectors, {each of length $t$, selected} from $N_i$. {We assume that the selected} 
$\{d_i^{j}, \ i \in [m] , j \in [R_i] \}$ are linearly independent {-- see our assumption in Proposition~\ref{thm:SecureAch_combi}.
We can write the messages decoded at} 
destination $D_i$, denoted by $\hat{W_i}$, {as follows}
	\begin{align}
	{\hat{W}_i} & = \begin{bmatrix} X_{e_1} & X_{e_2} & \ldots & X_{e_t}  \end{bmatrix} \begin{bmatrix}
	d_i^{(1)} & d_i^{(2)} & \ldots & d_i^{(R_i)}
	\end{bmatrix}.
	\end{align}
{We can stack all the decoded messages at the $m$ destinations together and obtain}
	\begin{align*}
&{\begin{bmatrix} \hat{W}_1 & \hat{W}_2 & \ldots & \hat{W}_m \end{bmatrix} }
\\	 & =  {\begin{bmatrix} X_{e_1} & X_{e_2} & \ldots & X_{e_t}  \end{bmatrix}} \left[ \begin{array}{ccccccccccc}
	d_1^{(1)} & \ldots & {d_1^{(R_1)}} & d_2^{(1)} & \ldots & d_2^{(R_2)} & \ldots & d_m^{(1)} & \ldots & d_m^{(R_m)}
	\end{array}\right] \\
	& \stackrel{\rm{(a)}}{=}  {\begin{bmatrix} W_1^T \\ W_2^T \\ \vdots \\ W_m^T \\ K^T \end{bmatrix}^T} \begin{bmatrix}
	E^T \\ V
	\end{bmatrix}\begin{bmatrix}
	d_1^{(1)} &  \ldots & {d_1^{(R_1)}} & d_2^{(1)} & \ldots & d_2^{(R_2)} & \ldots & d_m^{(1)} & \ldots & d_m^{(R_m)}
	\end{bmatrix} \\
	&  = \left[ {\begin{array}{c} W_1^T \\ W_2^T \\ \vdots \\ W_m^T \\ K^T \end{array}} \right]^T  \left[\begin{array}{c}
	E^T \left[ \begin{array}{ccccccccccc}
	d_1^{(1)} & \ldots & {d_1^{(R_1)}} & d_2^{(1)} & \ldots & d_2^{(R_2)} & \ldots & d_m^{(1)} & \ldots & d_m^{(R_m)}
	\end{array}\right] \\ ~ \\ V \left[ \begin{array}{ccccccccccc}
	d_1^{(1)} & \ldots & {d_1^{(R_1)}} & d_2^{(1)} & \ldots & d_2^{(R_2)} & \ldots & d_m^{(1)} & \ldots & d_m^{(R_m)}
	\end{array}\right]
	\end{array}\right] \\
	& \stackrel{{\rm{(b)}}}{=} \left[ {\begin{array}{c} W_1^T \\ W_2^T \\ \vdots \\ W_m^T \\ K^T \end{array}} \right]^T  \left[\begin{array}{c}
	E^T \left[ \begin{array}{ccccccccccc}
	d_1^{(1)} & \ldots & {d_1^{(R_1)}} & d_2^{(1)} & \ldots & d_2^{(R_2)} & \ldots & d_m^{(1)} & \ldots & d_m^{(R_m)}
	\end{array}\right] \\ ~ \\ {\mathbf{0}_{t \times \sum_{i=1}^m R_i}}
	\end{array}\right] \\
	& = \left[ {\begin{array}{c} W_1^T \\ W_2^T \\ \vdots \\ W_m^T \end{array}}\right]^T E^T {\underbrace{\left[ \begin{array}{ccccccccccc}
	d_1^{(1)} & \ldots & d_1^{(R_1)} & d_2^{(1)} & \ldots & d_2^{(R_2)} & \ldots & d_m^{(1)} & \ldots & d_m^{(R_m)}
	\end{array}\right]}_{D}}, 
	\end{align*}
{where the equality in $\rm{(a)}$ follows by using the definition in~\eqref{eq:enc} and the equality in $\rm{(b)}$ is due to the fact that} the vectors in $\{d_i^{(j)}\}$ are in the null space $N_i$, which is contained inside the null space of $V$. 
{Since we assumed that the selected $\{d_i^{j}, \ i \in [m] , j \in [R_i] \}$ are linearly independent, then this implies that the matrix $D$ has rank $\sum_{i=1}^m R_i$. Then, since
$E^T$ is a matrix of dimension $\sum_{i=1}^m R_i \times t$,} {one} can always find a matrix $E^T$ such that, 
	\begin{align*}
	{ E^T D = \mathbf{I}_{\sum_{i=1}^m R_i}.}
	\end{align*}
	Thus, we get,	
	\begin{align*}
	{\begin{bmatrix} \hat{W}_1 & \hat{W}_2 & \ldots & \hat{W}_m \end{bmatrix} }    & = {\begin{bmatrix} {W}_1 & {W}_2 & \ldots & {W}_m \end{bmatrix} }.
	\end{align*}
{This concludes the proof that there exists a choice of the matrix $E$ in~\eqref{eq:enc}, which ensures that all the destinations reliably decode their intended messages.}

%% file: 2CombNetAchMC.tex
\section{Proof that the rate region in~\eqref{eq:Cap2CombNet} is securely achievable}
\label{app:2CombNetAch}

{In order to show that the rate region in~\eqref{eq:Cap2CombNet} is securely achievable, we use a two-step proof. 
First, we determine all the feasible $(R_1,R_2)$ pairs that can be selected from the null space $N_i$, such that the assumption in Proposition~\ref{thm:SecureAch_combi} is satisfied, namely such that the $R_1+R_2$ selected vectors are linearly independent. Hence, the result in Proposition~\ref{thm:SecureAch_combi} implies that} any point in this region is achievable. 
Then, we prove that {the convex hull of these feasible rate pairs is indeed the region in~\eqref{eq:Cap2CombNet}. The proposition below represents the first step of our proof.} 
\begin{prop}
	\label{prop:pick_ind_vec_2}
	The convex hull of all { $(R_1, R_2)$ rate pairs} such that we can {select} $R_1$ vectors from $N_1$ and $R_2$ vectors from $N_2$, with all of these vectors being linearly independent, is given by the following region
	\begin{align*}
	 {R_1} & \leq \text{dim}(N_1), \\
	 {R_2} & \leq \text{dim} (N_2), \\
	{R_1 + R_2} & \leq \text{dim} (N_1 + N_2),
	\end{align*}
	where $+$ denotes the sum of {subspaces $N_1$ and $N_2$.}
\end{prop}

\begin{proof}
	
	We {first show} that we can select $\text{dim} (N_1)$ vectors from $N_1$ and $\text{dim} (N_1+ N_2) - \text{dim} (N_1)$ vectors from $N_2$, such that {all these} vectors are linearly independent. {We have:}
	\begin{itemize}
		\item $N_1$ is of dimension $\text{dim}(N_1)$ {and so we select} $\text{dim}(N_1)$ independent vectors from this space. One {feasible choice consists of selecting} the basis of {the} subspace $N_1$. Thus, $R_1 = \text{dim}(N_1)$.
		\item { In the basis of $N_1+N_2$ there are $\text{dim} (N_1 + N_2)$ independent vectors. Moreover,} note that the basis of $N_1 + N_2$ is a subset of the basis of $N_1$ union with the basis on $N_2$. So we can {select} vectors from the basis of $N_2$ as long as we {select an} independent vector. {Thus, we can select} $\text{dim} (N_{1} + N_2) - \text{dim} (N_{1})$ vectors from $N_{2}$, i.e., $R_{2} = \text{dim} (N_{1} + N_2) - \text{dim} (N_1) $.
	\end{itemize}
	
	By symmetry, one can also {first select} $\text{dim}(N_2)$ vectors from the null space $N_2$ and {then} $\text{dim}(N_1+N_2) - \text{dim}(N_2)$ vectors from the null space $N_1$. {For this case, one would get $R_2 = \text{dim} (N_2)$ and $R_1 = \text{dim} (N_{1} + N_2) - \text{dim} (N_2)$.} This completes the proof {since} these are the only non-trivial corner points for the region given in Proposition~\ref{prop:pick_ind_vec_2}.
\end{proof}

	We {now} prove that the region given in Proposition~\ref{prop:pick_ind_vec_2} {coincides with} the region given in {Proposition~\ref{thm:capacity}. We start by noting that we can rewrite the rate region in~\eqref{eq:Cap2CombNet} as}
	\begin{align*}
	{R_1} & \leq \left[ \left| \mathcal{M}_1 \right| - k \right ]^+ , \\
	{R_2} & \leq \left [ \left| \mathcal{M}_2 \right| - k \right]^+ , \\
	{R_1 + R_2} & \leq \min\left(\left[\left| \mathcal{M}_1 \right| - k \right]^+ + \left[\left| \mathcal{M}_2 \right| - k \right]^+, \left[\left| \mathcal{M}_1 \cup \mathcal{M}_2  \right| - k \right]^+ \right). 
	\end{align*}
Since, as {we have} proved in~\eqref{eq:dim_N_i}, $\text{dim} (N_i) = \left [ \left| \mathcal{M}_i \right| - k \right]^+, \ \forall i \in [m] $, {then} we only need to show that 
$$ \text{dim} (N_1 + N_2) = \min\left(\left [ \left| \mathcal{M}_1 \right| - k \right ]^+ + \left [ \left| \mathcal{M}_2 \right| - k \right]^+, \left [ \left| \mathcal{M}_1 \cup \mathcal{M}_2  \right| - k \right]^+ \right). $$  
{The dimension of the sum of two subspaces can be computed as}
$$\text{dim}(N_1+N_2) = \text{dim}(N_1) + \text{dim}(N_2) - \text{dim}(N_1 \cap N_2).$$
{Thus, we now need to compute $\text{dim}(N_1 \cap N_2)$. We note that} $N_1 \cap N_2$ is the null space of {the matrix
\begin{align*}
V^{\star}_{1,2} = \begin{bmatrix}V \\ C  \end{bmatrix}, \ \text{where} \ C = \begin{bmatrix} C_1 \\ C_2\end{bmatrix},
\end{align*}
with $C_1$ and $C_2$ being defined as in~\eqref{eq:vandermonde_i}. Moreover,}
there will be $t-\left| \mathcal{M}_1 \cap \mathcal{M}_2  \right|$ distinct rows {in $C$,} and following the argument based on $V$ being the generator matrix of a $(t,k,t-k+1)$ MDS code, then the number of independent rows in {$V^{\star}_{1,2}$ is} $\min(t, k + t- \left| \mathcal{M}_1 \cap \mathcal{M}_2  \right| )$. {Thus, 
\begin{align*}
\text{dim}(N_1 \cap N_2) = t - \min(t, k + t- \left| \mathcal{M}_1 \cap \mathcal{M}_2  \right| ) = \left [ \left| \mathcal{M}_1 \cap \mathcal{M}_2  \right| - k \right]^+,
\end{align*}
which leads to}
	\begin{align*}
	\text{dim}(N_1+N_2)  & = \left [ \left| \mathcal{M}_1 \right| - k \right ]^+ + \left [ \left| \mathcal{M}_2 \right| - k \right ]^+ - \left [ \left| \mathcal{M}_1 \cap \mathcal{M}_2  \right| - k \right]^+ \\
	& = \min\left( \left(\left| \mathcal{M}_1 \right| - k \right)^+ + \left(\left| \mathcal{M}_2 \right| - k \right)^+, \left(\left| \mathcal{M}_1 \cup \mathcal{M}_2 \right| - k \right)^+  \right).
	\end{align*}
	The last equality can be verified by considering all {possible} four cases, namely: (1) $\left| \mathcal{M}_1 \right| \geq k, \left| \mathcal{M}_2 \right| \geq k$, (2) $\left| \mathcal{M}_1 \right| < k, \left| \mathcal{M}_2 \right| \geq k$, (3) $\left| \mathcal{M}_1 \right| \geq k, \left| \mathcal{M}_2 \right| < k$ and (4) $\left| \mathcal{M}_1 \right| < k, \left| \mathcal{M}_2 \right| < k$. 
{This concludes the proof that the rate region in~\eqref{eq:Cap2CombNet} is securely achievable.}
	

%% file: single_source_unsecureMC.tex
\section{Unsecure capacity for single source multiple unicast traffic }
\label{app:single_source_unsecure}
We here give the proof of Lemma~\ref{thm:UnCapacity} (originally proved in~\cite[Theorem 9]{KoetterTON2003}). 
{We start by noting that, by setting $k=0$ in the outer bound in~\eqref{eq:SecureOB}, we readily obtain the rate region in~\eqref{eq:UnsecureOB}. It therefore follows that~\eqref{eq:UnsecureOB} is an outer bound on the capacity region of a multiple unicast network with single source and $m$ destinations. We now prove that the region in~\eqref{eq:UnsecureOB} is also achievable.}
Assume that a rate $m$-tuple $(R_1, R_2, \ldots, R_m)$ satisfies the constraint in~\eqref{eq:UnsecureOB}. 
We now prove that this $m$-tuple is achievable. 
Towards this end, from the original graph $\mathcal{G}$, we create a new directed acyclic graph $\mathcal{G}^\prime$ where a new node $D^\prime$ is connected to each $D_i, i \in {[m],}$ through an edge $\mathcal{E}_i^\prime$ of capacity $R_i$. 
It is not difficult to see that in $\mathcal{G}^\prime$, the min-cut capacity between $S$ and $D^\prime$ is $\sum\limits_{i=1}^m R_i$.  
This can be explained as follows. 
Suppose that the min-cut from $S$ to $D^\prime$, in addition to a subset of $\mathcal{E}$ (i.e., the set of edges in the original $\mathcal{G}$), also contains some edges $\mathcal{E}_{\mathcal{J}}^\prime$, with $\mathcal{J} \subseteq {[m]}$.
This clearly implies that the subset of edges from $\mathcal{E}$ should form a cut between source $S$ and {$D_{[m]\setminus \mathcal{J}}$,} otherwise we would not have a cut between $S$ and $D^\prime$.
Thus, the min-cut has a capacity of at least { $\sum\limits_{i \in J} R_i + M_{\left \{D_{[m]\setminus \mathcal{J}}\right\}}$} and, since $\sum\limits_{i \in {[m]}\setminus \mathcal{J}} R_i \leq M_{\left \{D_{{[m]}\setminus \mathcal{J}} \right \}}$ (this follows from the outer bound proved above), the min-cut has a capacity of at least {$\sum\limits_{i=1}^m R_i$.}
Then, since the set $\mathcal{E}_{{[m]}}^\prime$ is a cut of capacity {$\sum\limits_{i=1}^m R_i$,} it follows that the min-cut has a capacity of at most $\sum\limits_{i}^m R_i$.
This implies that the min-cut capacity between  $S$ and $D^\prime$ in $\mathcal{G}^\prime$ is $\sum\limits_{i=1}^m R_i$.
With this, the achievability of the rate $m$-tuple $(R_1, R_2, \ldots, R_m)$ that satisfies the constraint in~\eqref{eq:UnsecureOB} directly follows from the max-flow min-cut theorem.
Indeed, since one can communicate a total information of {$\sum\limits_{i=1}^m R_i$} from $S$ to $D^\prime$ in $\mathcal{G}^\prime$, then this is possible only if an amount $R_i$ of information flows through $D_i, i \in {[m],}$ in $\mathcal{G}$. 	
This concludes the proof of Lemma~\ref{thm:UnCapacity}.
Notice that in order to transmit $\sum\limits_{i=1}^m R_i$ message packets from $S$ to $D^\prime$ (single unicast session) network coding is not needed. Thus, there is no need of coding operations to characterize the capacity region of a network with single source and multiple destinations.

%% file: OtherNetInstMC.tex
\section{Secure Capacity Results on Other Instances of Multiple Unicast Traffic}

\subsection{Erasure Networks}
\label{app:OtherNetInstEras}	
{We here report the secure capacity region results that we derived in~\cite{agarwal2016secure} for the three networks in Fig.~\ref{fig:networks}. In particular, the secure capacity regions can be found as the solution of some feasibility programs. We refer an interested reader to~\cite{agarwal2016secure} for the complete proof of these results.}
\begin{prop}
		\label{thm:capY}
		The secure capacity region of the { Y-network} in Fig.~\ref{fig:networks}(a) is given by
		\begin{subequations}
			\label{eq:capY}
			\begin{align}
			k_j  \geq  R_j \frac{1-\delta_{j \text{E}}}{1-\delta_j \delta_{j \text{E}}},   \ j \in [2], \label{eq:capY1}
			\\k_3  \geq  (R_1+R_2) \frac{1-\delta_{3\text{E}}}{1-\delta_3 \delta_{3\text{E}}},   \label{eq:capY2}
			\\ \frac{R_j}{1-\delta_j} + \frac{k_j}{(1-\delta_j)\delta_{j \text{E}}}  \leq  1,  \ j \in [2],   \label{eq:capY3} 
			\\ \frac{R_1+R_2}{1-\delta_3} + \frac{k_3}{(1-\delta_3)\delta_{3\text{E}}}  \leq  1 ,  \label{eq:capY4}
			 \\ k_3 \leq  \left ( \frac{k_1}{\delta_{1\text{E}}}+ \frac{k_2}{\delta_{2\text{E}}} \right ) \frac{(1-\delta_3) \delta_{3\text{E}}}{1-\delta_3 \delta_{3\text{E}}},  \label{eq:capY5}
			\\ R_i,k_j \geq 0, \ i \in [1:2], \ j \in [3], \label{eq:capY6}
			\end{align}
		\end{subequations}
{where: (i)} the first and the second {constraints} ensure that {enough keys are generated, i.e., the number of generated keys is larger than the amount of information received by the adversary; (ii)} the third and the fourth {inequalities} are time constraints ensuring {that the length of the key generation phase plus the length of the message sending phase do not exceed the total available time; (iii) finally,} the fifth constraint follows {since node $I$ has zero randomness and so the key that it can generate is constrained by the randomness received from $S_1$ and $S_2$.}
	\end{prop}

		\begin{prop}
			\label{thm:capRY}
			The secure capacity region of the {RY-network} in Fig.~\ref{fig:networks}(b) is given by
			\begin{subequations}
				\label{eq:capRY}
				\begin{align}
			k_3 + e \frac{(1-\delta_3) \delta_{3 \text{E}}}{1-\delta_3 \delta_{3\text{E}}}  \geq  (R_1+R_2) \frac{1-\delta_{3 \text{E}}}{1-\delta_3 \delta_{3 \text{E}}}, \label{eq:capRY1}
			\\ k_j  \geq  R_j \frac{1-\delta_{j \text{E}}}{1-\delta_j \delta_{j \text{E}}},    \ j \in [2], \label{eq:capRY2} 
			\\ \frac{R_1+R_2}{1-\delta_3} + \frac{k_3}{(1-\delta_3)\delta_{3 \text{E}}} + \frac{e}{1-\delta_3}  \leq  1, \label{eq:capRY3}  
			\\ \frac{R_j}{1-\delta_j} + \frac{k_j}{(1-\delta_j)\delta_{j \text{E}}}  \leq  1,  \  j \in [2], \label{eq:capRY4}
			\\  k_3  \leq (D_0 - e) \frac{(1-\delta_3) \delta_{3 \text{E}} }{1-\delta_3 \delta_{3 \text{E}}},\label{eq:capRY5} 
			\\  k_j  \leq  \left ( e + \frac{k_3}{\delta_{3 \text{E}}} \right ) \frac{(1-\delta_j) \delta_{j \text{E}}}{1-\delta_j \delta_{j \text{E}}} , \  j \in [2], \label{eq:capRY6}
			\\  R_i, e,k_j  \geq 0, \ i \in [1:2], \ j \in [3], \label{eq:capRY7}
				\end{align}
			\end{subequations}
{where: (i)} the first and the second {constraints} ensure that {enough keys are generated, i.e., the number of generated keys is larger than the amount of information received by the adversary; (ii)} the third and the fourth {inequalities} are time constraints ensuring {that the length of the key generation phase plus the length of the message sending phase do not exceed the total available time; (iii) finally, the fifth (respectively, sixth) constraint is due to the fact that the key that node $S$ (respectively, node $I$) can create is constrained by its limited randomness (respectively, the randomness that it receives from $S$).}
		\end{prop}

		\begin{prop}
			\label{thm:capX}
			The secure capacity region of the {X-network} in Fig.~\ref{fig:networks}(c) is given by
			\begin{subequations}
				\label{eq:capX}
				\begin{align}
				k_j  \geq  R_j \frac{1-\delta_{j \text{E}}}{1-\delta_j \delta_{j \text{E}}}, \ j \in [2], \label{eq:capX1}
				\\ k_3 + e \frac{(1-\delta_3) \delta_{3 \text{E}}}{1-\delta_3 \delta_{3\text{E}}} \geq  (R_1+R_2) \frac{1-\delta_{3 \text{E}}}{1-\delta_3 \delta_{3 \text{E}}}, \label{eq:capX2}
				\\ k_j  \geq R_{j-3} \frac{1-\delta_{j \text{E}}}{1-\delta_j \delta_{j \text{E}}} , \ j \in [4:5], \label{eq:capX3}
				\\	\frac{R_j}{1-\delta_j} + \frac{k_j}{(1-\delta_j)\delta_{j \text{E}}}    \leq  1, \  j \in [2], \label{eq:capX4}
				\\ \frac{R_{j-3}}{1-\delta_j} + \frac{k_j}{(1-\delta_j)\delta_{j \text{E}}}  \leq  1,  \ j \in [4:5], \label{eq:capX5}
				\\ \frac{R_1+R_2}{1-\delta_3} + \frac{k_3}{(1-\delta_3)\delta_{3 \text{E}}} + \frac{e}{1-\delta_3}  \leq  1, \label{eq:capX6} 
				\\  k_3  \leq  \left ( \frac{k_1}{\delta_{1\text{E}}}+  \frac{k_2}{\delta_{2\text{E}}} - e \right ) \frac{(1-\delta_3) \delta_{3 \text{E}}  }{1-\delta_3 \delta_{3 \text{E}}}, \label{eq:capX7}  
				\\ k_j  \leq   \left (e+ \frac{k_3}{\delta_{3\text{E}}} \right ) \frac{(1-\delta_j) \delta_{j \text{E}}}{1-\delta_j \delta_{j \text{E}}} , \  j \in [4:5], \label{eq:capX8}
				\\  R_i, e,k_j \geq 0, \ i \in [1:2], \ j \in [5], \label{eq:capX9}				
				\end{align}
			\end{subequations}
{where: (i)} the first, second and third {constraints} ensure that {enough keys are generated, i.e., the number of generated keys is larger than the amount of information received by the adversary; (ii)} the fourth, fifth and sixth {inequalities} are time constraints ensuring {that the length of the key generation phase plus the length of the message sending phase do not exceed the total available time; (iii) finally, the seventh and the eight constraints are due to the fact that the key that a node can create is constrained by the randomness that it receives from previous nodes.}
		\end{prop}

\begin{table}[ht]

\begin{center}

	\begin{tabular}{ 
			|>{\centering\arraybackslash}m{0.05\textwidth}|
			|>{\centering\arraybackslash}m{0.47\textwidth}|
			>{\centering\arraybackslash}m{0.35\textwidth}|
		}
		\hline 
		&  \textbf{Unsecure capacity region} &  \textbf{Secure capacity region} \\
		\hline
		\begin{turn}{90}{\bf{Butterfly network~1}}\end{turn}
		&
		{\begin{subequations}
				\label{eq:but1cap}
				\begin{align}
				R_1  & \leq \min \left \{\mathsf{C}_1,\mathsf{C}_3,\mathsf{C}_7 \right \}, \label{eq:but1capc1} \\
				R_2  & \leq \min \left \{\mathsf{C}_2,\mathsf{C}_3,\mathsf{C}_6 \right \} \label{eq:but1capc2},
				\\  R_1 + R_2 & \leq \mathsf{C}_3 + \min \left\{ \mathsf{C}_4, \mathsf{C}_5 \right\}\label{eq:but1capc3}.
				\end{align} 
		\end{subequations}} 
		&
		{\begin{subequations}
				Secure communication is not possible.
		\end{subequations}}
		\\
		\hline
		\begin{turn}{90}{\bf{Single source}}\end{turn}
		&
		{\begin{subequations}
				\label{eq:CScap}
				\begin{align}
				R_1  & \!\leq \!\mathsf{C}_5 \!+\! \min \left \{ \mathsf{C}_1	\!+\!\mathsf{C}_2,\mathsf{C}_3,\mathsf{C}_7 \right \},
				\\ R_2  & \!\leq\! \mathsf{C}_4\!+\! \min \left \{ \mathsf{C}_1\!+\!\mathsf{C}_2,\mathsf{C}_3,\mathsf{C}_6 \right \},
				\\  R_1 + R_2  & \!\leq\! \mathsf{C}_4 \!+\! \mathsf{C}_5 + \min \left \{ \mathsf{C}_1\!+\!\mathsf{C}_2,\mathsf{C}_3,\mathsf{C}_6\!+\!\mathsf{C}_7 \right \}.
				\end{align} 
		\end{subequations}} 
		&
		{\begin{subequations}
				\label{eq:CScapSec}
				\begin{align}
				 R_1 & \leq \min \left \{\mathsf{C}_5,\mathsf{C}_1\!+\!\mathsf{C}_2,\mathsf{C}_3,\mathsf{C}_7 \right \}, \label{eq:bfcs_seq_ach_a}
				\\	 R_2 & \leq \min \left \{\mathsf{C}_4,\mathsf{C}_1\!+\!\mathsf{C}_2,\mathsf{C}_3,\mathsf{C}_6\right \}. \label{eq:bfcs_seq_ach_b}
				\end{align}
		\end{subequations}}
		\\
		\hline
		\begin{turn}{90}{\bf{Single destination}}\end{turn}
		&
		{\begin{subequations}
				\label{eq:CDcap}
				\begin{align}
				R_1 & \!\leq\! \mathsf{C}_4 \!+\! \min \left \{ \mathsf{C}_1,\mathsf{C}_3,\mathsf{C}_6\!+\!\mathsf{C}_7\right \},
				\\ R_2 & \!\leq\! \mathsf{C}_5\!+\! \min \left \{\mathsf{C}_2,\mathsf{C}_3,\mathsf{C}_6\!+\!\mathsf{C}_7 \right \},
				\\ R_1 + R_2 &   \!\leq\! \mathsf{C}_4 \!+\! \mathsf{C}_5  + \min \left \{\mathsf{C}_1\!+\!\mathsf{C}_2,\mathsf{C}_3,\mathsf{C}_6\!+\!\mathsf{C}_7 \right \}.
				\end{align}
		\end{subequations}} 
		&
		{\begin{subequations}
				\label{eq:CDcapSec}
				\begin{align}
				R_1 & \leq \min \left \{\mathsf{C}_1,\mathsf{C}_4 \right \}, \label{eq:bfcd_seq_ach_a}
				\\ R_2 & \leq \min \left \{ \mathsf{C}_2,\mathsf{C}_5\right \}, \label{eq:bfcd_seq_ach_b}
				\\  R_1 + R_2 & \leq \min \left \{ \mathsf{C}_3,\mathsf{C}_6+\mathsf{C}_7\right \}. \label{eq:bfcd_seq_ach_c}
				\end{align}
		\end{subequations}}
		\\
		\hline
		\begin{turn}{90}{\bf{Butterfly Nnetwork~2}}\end{turn}
		&
		{\begin{subequations}
				\label{eq:Butt2cap}
				\begin{align}
				R_1  & \leq \mathsf{C}_4 + \min \left \{\mathsf{C}_1,\mathsf{C}_3,\mathsf{C}_7 \right \},
				\\ R_2 & \leq \mathsf{C}_5+ \min \left \{ \mathsf{C}_2,\mathsf{C}_3,\mathsf{C}_6 \right \},
				\\  R_1 + R_2 & \leq \mathsf{C}_4 + \mathsf{C}_5+ \mathsf{C}_3.
				\end{align} 
		\end{subequations}} 
		&
		{\begin{subequations}
				\label{eq:Butt2capSec}
				\begin{align}
				R_1 & \leq  \min \left \{ \mathsf{C}_4,\mathsf{C}_1,\mathsf{C}_3,\mathsf{C}_7\right \}, \label{eq:bf2_seq_ach_a}
				\\ R_2 & \leq  \min \left \{ \mathsf{C}_5,\mathsf{C}_2,\mathsf{C}_3,\mathsf{C}_6\right \}, \label{eq:bf2_seq_ach_b}
				\\   R_1 + R_2 &  \leq \mathsf{C}_3. \label{eq:bf2_seq_ach_c}
				\end{align}
		\end{subequations}}
		\\
		\hline		
		\end{tabular}
	\end{center}
	\caption{{Unsecure and secure capacity regions for the networks in Fig.~\ref{fig:butterflies}.}}
	\label{table:MessaDescr}
\end{table}

\subsection{Arbitrary Edge Capacities}
\label{app:OtherNetArbrCap}	
{We here report the unsecure and secure capacity region results that we derived in~\cite{agarwal2016netcod} for the four networks in Fig.~\ref{fig:butterflies}. In particular, these results are shown in Table~\ref{table:MessaDescr}. We refer an interested reader to~\cite{agarwal2016netcod} for the complete proof of these results.}